\documentclass{article}
\usepackage{smartref} 
\usepackage{fullpage}
\usepackage{authblk}
\usepackage{amsthm}
\usepackage{paralist}
\usepackage{url}

\makeatletter\newenvironment{graybox}{%
   \begin{lrbox}{\@tempboxa}\begin{minipage}{\columnwidth}}{\end{minipage}\end{lrbox}%
   \colorbox{gray!25}{\usebox{\@tempboxa}}
}\makeatother
\usepackage{pgf}
\usepackage{tikz}
\usetikzlibrary{arrows,shapes,snakes,automata,backgrounds,petri}

\newcounter{mylinenum}
\addtoreflist{mylinenum}
\def\codeTabSpace{\hspace*{6mm}}
\newenvironment{code}%
{\begin{tabbing}%
\codeTabSpace \= \hspace*{72mm} \= \hspace*{33mm} \= \kill%
}%
{\end{tabbing}%
}
\newcounter{ind}
\newcommand{\n}{\addtocounter{ind}{5}\hspace*{5mm}}
\newcommand{\p}{\addtocounter{ind}{-5}\hspace*{-5mm}}
\newcommand{\nl}{\\\stepcounter{mylinenum}{\scriptsize \arabic{mylinenum}}\>\hspace*{\value{ind}mm}}
\newcommand{\ul}{\\\>\hspace*{\value{ind}mm}}
\newcommand{\bl}{\\[-1.5mm]\>\hspace*{\value{ind}mm}}
\newcommand{\firstline}{\stepcounter{mylinenum}{\scriptsize \arabic{mylinenum}}\>}

\newcommand{\comnospace}{\mbox{$\triangleright$}}
\newcommand{\com}{\mbox{\comnospace\ }}

\newcommand{\TRUE}{\mbox{\sc True}}
\newcommand{\FALSE}{\mbox{\sc False}}

\newcommand{\func}[1]{\mbox{\sc #1}}
\newcommand{\CAS}{\textbf{\func{CAS}}}

\newcommand{\tabtabcom}{\>\>\com}
\newcommand{\tabcom}{\>\com}

\newcommand{\key}{Key}
\newcommand{\info}{Descriptor}
\newcommand{\fldinfo}{pending}
\newcommand{\node}{Node}
\newcommand{\leaf}{Leaf}
\newcommand{\internal}{Internal}
\newcommand{\iinfo}{ReplaceFlag}
\newcommand{\dinfo}{PruneFlag}
\newcommand{\mk}{Mark}

\newcommand{\clean}{Clean}

\newcommand{\Isimple}{Simple insertion}
\newcommand{\Icomplex}{Sprouting insertion}
\newcommand{\Dsimple}{Simple deletion}
\newcommand{\Dcomplex}{Pruning deletion}

\newcommand{\cas}{CAS}

\newcommand{\mcas}{\textit{Mark~\cas}}

\newcommand{\ifcas}{\textit{Rflag~\cas}}
\newcommand{\dfcas}{\textit{Pflag~\cas}}
\newcommand{\iucas}{\textit{Runflag~\cas}}
\newcommand{\ducas}{\textit{Punflag~\cas}}
\newcommand{\bcas}{\textit{Backtrack~\cas}}
\newcommand{\ichildcas}{\textit{Rchild~\cas}}
\newcommand{\dchildcas}{\textit{Pchild~\cas}}

\newcommand{\typeof}{type}
\newcommand{\helpinsert}{\func{HelpReplace}}

\newcommand{\helpdelete}{\func{HelpPrune}}
\newcommand{\helpmarked}{\func{HelpMarked}}
\newcommand{\help}{\func{Help}}

\newcommand{\insertop}{\textsc{Insert}}
\newcommand{\deleteop}{\textsc{Delete}}

\newcommand{\true}{\textsc{True}}
\newcommand{\false}{\textsc{False}}
\newcommand{\nul}{\textsc{Null}}

\newcommand{\kst}{\mbox{$k$-ST}}

\newcommand{\codesize}{\footnotesize}

\newcommand{\ins}{\func{Insert}}
\newcommand{\del}{\func{Delete}}
\newcommand{\find}{\func{Find}}

\newcommand{\rquery}{\func{RangeQuery}}

\newcounter{ttheorem}
\newtheorem{obs}[ttheorem]{Observation}{\bfseries}{\itshape}
\newtheorem{lem}[ttheorem]{Lemma}{\bfseries}{\itshape}
\newtheorem{thm}[ttheorem]{Theorem}{\bfseries}{\itshape}
{\bfseries}{\itshape}
\newtheorem{defn}[ttheorem]{Definition}{\bfseries}{\itshape}

\tikzstyle{tri} = [regular polygon,regular polygon sides=3,draw=black!80,fill=blue!13,minimum size=16mm]
\tikzstyle{dashed square} = [rectangle,dotted,draw=black!80,minimum size=5mm]
\tikzstyle{square} = [rectangle,draw=black!80,fill=blue!13,minimum size=5.3mm]
\tikzstyle{ghost square} = [rectangle,draw=black!35,fill=blue!5,minimum size=5.3mm]
\tikzstyle{ghost path} = [draw=black!35,thick,-]
\tikzstyle{place} = [circle,thick,draw=blue!80,fill=blue!13,minimum size=7mm]
\tikzstyle{red square} = [square,draw=black!80,fill=red!13]
\tikzstyle{red place} = [place,draw=red!75,fill=red!20]
\tikzstyle{green square} = [square,draw=black!80,fill=green!13]
\tikzstyle{invisiplace} = [minimum size=6mm]
\tikzstyle{every path} = [draw=black!80,thick,-]
\tikzstyle{soft path} = [draw=gray!80,-]
\tikzstyle{big path} = [draw=blue!50,-,double]
\tikzstyle{every label} = [gray]
\tikzstyle{blacklabel} = [black]

\begin{document}

\makeatletter
%
%
\pgfdeclareshape{rectangle with diagonal fill}
{
    \inheritsavedanchors[from=rectangle]
    \inheritanchorborder[from=rectangle]
    \inheritanchor[from=rectangle]{north}
    \inheritanchor[from=rectangle]{north west}
    \inheritanchor[from=rectangle]{north east}
    \inheritanchor[from=rectangle]{center}
    \inheritanchor[from=rectangle]{west}
    \inheritanchor[from=rectangle]{east}
    \inheritanchor[from=rectangle]{mid}
    \inheritanchor[from=rectangle]{mid west}
    \inheritanchor[from=rectangle]{mid east}
    \inheritanchor[from=rectangle]{base}
    \inheritanchor[from=rectangle]{base west}
    \inheritanchor[from=rectangle]{base east}
    \inheritanchor[from=rectangle]{south}
    \inheritanchor[from=rectangle]{south west}
    \inheritanchor[from=rectangle]{south east}

    \inheritbackgroundpath[from=rectangle]
    \inheritbeforebackgroundpath[from=rectangle]
    \inheritbehindforegroundpath[from=rectangle]
    \inheritforegroundpath[from=rectangle]
    \inheritbeforeforegroundpath[from=rectangle]

    \behindbackgroundpath{%
        \pgfextractx{\pgf@xa}{\southwest}%
        \pgfextracty{\pgf@ya}{\southwest}%
        \pgfextractx{\pgf@xb}{\northeast}%
        \pgfextracty{\pgf@yb}{\northeast}%
        \ifpgf@diagonal@lefttoright
            \def\pgf@diagonal@point@a{\pgfpoint{\pgf@xa}{\pgf@yb}}%
            \def\pgf@diagonal@point@b{\pgfpoint{\pgf@xb}{\pgf@ya}}%
        \else
            \def\pgf@diagonal@point@a{\southwest}%
            \def\pgf@diagonal@point@b{\northeast}%
        \fi
        \pgfpathmoveto{\pgf@diagonal@point@a}%
        \pgfpathlineto{\northeast}%
        \pgfpathlineto{\pgfpoint{\pgf@xb}{\pgf@ya}}%
        \pgfpathclose
        \ifpgf@diagonal@lefttoright
            \color{\pgf@diagonal@top@color}%
        \else
            \color{\pgf@diagonal@bottom@color}%
        \fi
        \pgfusepath{fill}%
        \pgfpathmoveto{\pgfpoint{\pgf@xa}{\pgf@yb}}%
        \pgfpathlineto{\southwest}%
        \pgfpathlineto{\pgf@diagonal@point@b}%
        \pgfpathclose
        \ifpgf@diagonal@lefttoright
            \color{\pgf@diagonal@bottom@color}%
        \else
            \color{\pgf@diagonal@top@color}%
        \fi
        \pgfusepath{fill}%
    }
}

\newif\ifpgf@diagonal@lefttoright
\def\pgf@diagonal@top@color{white}
\def\pgf@diagonal@bottom@color{gray!30}

\def\pgfsetdiagonaltopcolor#1{\def\pgf@diagonal@top@color{#1}}%
\def\pgfsetdiagonalbottomcolor#1{\def\pgf@diagonal@bottom@color{#1}}%
\def\pgfsetdiagonallefttoright{\pgf@diagonal@lefttorighttrue}%
\def\pgfsetdiagonalrighttoleft{\pgf@diagonal@lefttorightfalse}%

\tikzoption{diagonal top color}{\pgfsetdiagonaltopcolor{#1}}
\tikzoption{diagonal bottom color}{\pgfsetdiagonalbottomcolor{#1}}
\tikzoption{diagonal from left to right}[]{\pgfsetdiagonallefttoright}
\tikzoption{diagonal from right to left}[]{\pgfsetdiagonalrighttoleft}
\makeatother

\newcommand{\trevor}[1]{\textbf{[[[#1--Trevor]]]}}

\title{Range Queries in Non-blocking $k$-ary Search Trees}
\author{Trevor Brown$^1$, Hillel Avni$^2$
}
\affil{
$^1$Dept. of Computer Science, IST Austria\\%
me@tbrown.pro\\%
$^2$Huawei Technologies, European Research Institute\\%
hillel.avni@huawei.com
}
\date{}
\maketitle

\begin{abstract}
We present a linearizable, non-blocking $k$-ary search tree (\kst)
that supports fast searches and range queries.
Our algorithm uses single-word compare-and-swap (CAS) operations,
and tolerates any number of crash failures.
Performance experiments show that, for workloads containing small range queries, our \kst\ significantly outperforms other algorithms which support these operations, and rivals the performance of a leading concurrent skip-list, which provides range queries that cannot always be linearized.

%
\end{abstract}

\section{Introduction and Related Work}

The ordered set abstract data type (ADT) represents a set of keys drawn from an ordered universe, and supports three operations: \func{Insert}$(key)$, \func{Delete}$(key)$, and \func{Find}$(key)$.  
We add to these an operation \func{RangeQuery}$(a,b)$, where $a\le b$, which returns all keys in the closed interval $[a,b]$.
This is useful for various database applications.

Perhaps the most straightforward way to implement this ADT is to employ software transactional memory (STM) \cite{Shavit:1995:STM:224964.224987}. 
STM 
allows a programmer to specify that certain blocks of code should be executed atomically, relative to one another.
Recently, several fast binary search tree algorithms using STM have been introduced \cite{DBLP:conf/podc/BronsonCCO10,DBLP:conf/opodis/AfekAS11}.  Although they offer good performance for \func{Insert}s, \func{Delete}s and \func{Find}s, they achieve this performance, in part, by carefully limiting the amount of data protected by their transactions. 
However, since computing a range query means protecting all keys in the range from change during a transaction, STM techniques presently involve too much overhead to be applied to this problem.

Another simple approach is to lock the entire data structure, and compute a range query while it is locked. 
One can refine this technique by using a more fine-grained locking scheme, so that only part of the data structure needs to be locked to perform an update or compute a range query.
For instance, in leaf-oriented trees, where all keys in the set are stored in the leaves of the tree, updates to the tree can be performed by local modifications close to the leaves.  Therefore, it is often sufficient to lock only the last couple of nodes on the path to a leaf, rather than the entire path from the root.
However, as was the case for STM, a range query can only be computed if every key in the range is protected, so typically every node containing a key in the range must be locked.

Persistent data structures \cite{Okasaki:1998:PFD:280586} offer another approach. 
The nodes in a persistent data structure are immutable, so updates create new nodes, rather than modifying existing ones. 
In the case of a persistent tree, a change to one node involves recreating the entire path from the root to that node.
After the change, the data structure has a new root, and the old version of the data structure remains accessible (via the old root).
Hence, it is trivial to implement range queries in a persistent tree. 
However, significant downsides include contention at the root, and the duplication of many nodes during updates.

Brown and Helga \cite{BH11:opodis} presented a \kst\ in which each internal node has $k$ children, and each leaf contains up to $k-1$ keys.  For large values of $k$, this translates into an algorithm which minimizes cache misses and benefits from processor pre-fetching mechanisms.
In some ways, the \kst\ is similar to a persistent data structure.
The keys of a node are immutable, but the child pointers of a node can be changed.
The structure is also leaf-oriented, meaning that all keys in the set are stored in the leaves of the tree.
Hence, when an update adds or removes a key from the set, the leaf into which the key should be inserted, or from which the key should be deleted, is simply replaced by a new leaf.
Since the old leaf's keys remains unmodified, range queries using this leaf need only check that it has not been replaced by another leaf to determine that its keys are all in the data structure.
To make this more efficient, we modify this structure by adding a \textit{tag} field to each leaf, which is set just before the leaf is replaced.

Braginsky and Petrank \cite{Braginsky:2012:LB:2312005.2312016} presented a non-blocking B$^+$tree, another search tree of large arity.  However, whereas the \kst's nodes have immutable keys, the nodes of Braginsky's B$^+$tree do not.  Hence, our technique for performing range queries cannot be efficiently applied to their data structure.

Snapshots offer another approach for implementing range queries.
If we could quickly take a snapshot of the data structure, then we could simply perform a sequential range query on the result.
The snapshot object is a vector $V$ of data elements supporting two operations: \func{Update}$(i, val)$, which atomically sets $V_i$ to $val$, and \func{Scan}, which atomically reads and returns all of the elements of $V$.
\func{Scan} can be implemented by repeatedly performing a pair of \func{Collect}s (which read each element of $V$ in sequence and return a new vector containing the values it read) until the results of the two \func{Collect}s are equal \cite{Afek:1993:ASS:153724.153741}.
Attiya, et~al. \cite{Attiya:2008:PSO:1378533.1378591} introduced \textit{partial snapshots}, offering a modified \func{Scan}$(i_1, i_2, ..., i_n)$ operation which operates on a subset of the elements of $V$.
Their construction requires both CAS and fetch-and-add.

Recently, two high-performance tree offering $O(1)$ time snapshots have been published.
Both structures use a lazy copy-on-write scheme that we now describe.

Ctrie is a non-blocking concurrent hash trie due to Prokopec et~al. \cite{PBBO12:ppopp}.
Keys are hashed, and the bits of these hashes are used to navigate the trie.
To facilitate the computation of fast snapshots, a sequence number is associated with each node in the data structure.
Each time a snapshot is taken, the root is copied and its sequence number is incremented.
An update or search in the trie reads this sequence number $seq$ when it starts and, while traversing the trie, it duplicates each node whose sequence number is less than $seq$.
The update then performs a variant of a double-compare-single-swap operation to atomically change a pointer while ensuring the root's current sequence number matches $seq$.
Because keys are ordered by their hashes in the trie, it is hard to use Ctrie to efficiently implement range queries.
To do so, one must iterate over all keys in the snapshot.

The second tree, Snap, is a lock-based AVL tree due to Bronson et~al. \cite{BCCO10:ppopp}.
Whereas Ctrie added sequence numbers, Snap \textit{marks} each node to indicate that it should no longer by modified.
Updates are organized into \textit{epochs}, with each epoch represented by an object in memory containing a count of the number of active updates belonging to that epoch.
A snapshot marks the root node, ends the current epoch, and blocks further updates from starting until all updates in the current epoch finish.
Once updates are no longer blocked, they copy and mark each node they see whose parent is marked.
%
%
Like Ctrie, this pushes work from snapshots onto subsequent updates.  If these snapshots are used to compute small range queries, this may result in excessive duplication of unrelated parts of the structure.



If we view shared memory as a contiguous array, then our range queries are similar to \textit{partial snapshots}.
We implement two optimizations specific to our data structure.
First, when we traverse the tree to perform our initial \func{Collect}, we need only read a pointer to each \textit{leaf} that contains a key in the desired range (rather than reading each key).
This is a significant optimization when $k$ is large, e.g., 64.
Second, in the absence of contention, instead of performing a second \func{Collect}, 
we can simply check the \textit{tag} field of each leaf found by the first \func{Collect}.
As a further optimization, range queries can return a sequence of leaves, rather than copying their keys into an auxiliary data structure.
\\




\noindent Contributions of this work:
\begin{enumerate}[$\bullet$]
	\item{
		We present a new, provably correct data structure, and demonstrate experimentally that, for two very different sizes of range queries, it significantly outperforms data structures offering $O(1)$ time snapshots.  In many cases, it even outperforms a non-blocking skip-list, whose range queries cannot always be linearized.
	
	}
	\item{
		We contribute to a better understanding of the performance limitations of the $O(1)$ time snapshot technique for this application.
	}
\end{enumerate}


The rest of the paper is structured as follows.
In Section~\ref{sec-kst}, we describe the data structure, how updates are performed,
and our technique for computing partial snapshots of the nodes of the tree.
We give the details of how range queries are computed from these partial snapshots in Section~\ref{sec-code}.
Correctness and progress are proved in
Section~\ref{sec-correctness}.
Experimental results are presented in Section~\ref{sec-exp}.
Finally, we conclude in Section~\ref{sec-future}.

\section{Basic Operations} \label{sec-kst}

The \kst\ is a linearizable, leaf-oriented search tree in which each internal node has $k$ children and each leaf contains up to $k-1$ keys.
Its non-blocking operations, \find, \ins\ and \del,
implement the set ADT.
\func{Find}$(key)$ returns \true\ if $key$ is in the set, and \false\ otherwise. 
If $key$ is not already in the set, then \func{Insert}$(key)$ adds $key$ and returns \true. Otherwise it returns \false. 
If $key$ is in the set, then \func{Delete}$(key)$ removes $key$ and returns \true. Otherwise it returns \false. 
%
%
The \kst\ can be extended to implement the dictionary ADT,
in which a value is associated with each key
(described in the technical report for \cite{BH11:opodis}).
Although a leaf-oriented tree occupies more space
than a node-oriented tree (since it contains up to twice as many nodes),
the expected depth of a key or value is only marginally higher,
since more than half of the nodes are leaves.
Additionally, the fact that all updates in a leaf-oriented tree
can be performed by a local change near the leaves dramatically
simplifies the task of proving that updates do not interfere with one another.

\subsubsection*{Basic \kst\ operations}

\begin{figure}[t!]
\begin{center}
	\newcommand{\figsize}{0.68}
	\begin{minipage}{0.48\textwidth}
		\centering
		\noindent\textbf{Sprouting Insertion} \\
		\begin{tikzpicture}[node distance=1.2cm,>=stealth',bend angle=10,auto]
		  \begin{scope}
		  	\def\dist{0.5}
		  	\def\pad{0.5*\dist}
		  	\def\spaces{6}
		  	\def\offy{0}
			\foreach \xa/\ya/\xb/\yb/\p in {%
				1*\dist+\spaces/+0/-1*\dist+\pad+\spaces/-2/-1,%
				2*\dist+\spaces/+0/+1*\dist+0.333*\pad+\spaces/-2/-1,%
				3*\dist+\spaces/+0/+3*\dist-0.333*\pad+\spaces/-2/-1,%
				3*\dist+\spaces/+0/+5*\dist-\pad+\spaces/-2/+1,%
				-4*\dist/+2/+0/+0/+1
			} {
				\draw [-] (\xa*\dist+\p*\pad,\ya*\dist-\pad+\offy) --
						  (\xb*\dist,\yb*\dist+\pad+\offy);
			}
			\draw [-] (\spaces-\dist-4*\dist*\dist+1*\pad,2*\dist-\pad+\offy)
				edge [->, bend right] (\spaces-8*\dist*\dist,0*\dist+\pad+\offy);
			
		  	\foreach \num/\x/\y/\type in {%
		  		  /-5*\dist/+2/square,%
		  		  /-3*\dist/+2/square,%
		  		  /-1*\dist/+2/square,%
		  		  /\spaces-\dist-5*\dist/+2/square,%
		  		  /\spaces-\dist-3*\dist/+2/square,%
		  		  /\spaces-\dist-1*\dist/+2/square,%
		  		 a/\spaces-03*\dist/+0/ghost square,%
		  		 c/\spaces-01*\dist/+0/ghost square,%
		  		 d/\spaces+01*\dist/+0/ghost square,%
		  		 a/-01/+0/square,%
		  		 c/+00/+0/square,%
		  		 d/+01/+0/square,%
		  		 b/1*\dist+\spaces/+0/square,%
		  		 c/2*\dist+\spaces/+0/square,%
		  		 d/3*\dist+\spaces/+0/square,%
		  		 a/-1*\dist+\pad+\spaces/-2/square,%
		  		 b/+1*\dist+0.333*\pad+\spaces/-2/square,%
		  		 c/+3*\dist-0.333*\pad+\spaces/-2/square,%
		  		 d/+5*\dist-\pad+\spaces/-2/square%
		  	} {
				\draw (\x*\dist,\y*\dist+\offy) node [\type] {\num};
		  	}
		  	\draw [-to,thick,snake=snake,segment amplitude = 0.4mm,%
		  		segment length=2mm,line after snake=1mm]
		  		(3*\dist,\offy) -- node[text centered,xshift=0cm]
		  		{\insertop$(b)$} (-\dist+\dist*\spaces,\offy);
		  \end{scope}
		\end{tikzpicture}

        \vspace{3mm}
		\noindent\textbf{Simple Insertion} \\
		\begin{tikzpicture}[node distance=1.2cm,>=stealth',bend angle=10,auto]
		  \begin{scope}
		  	\def\dist{0.5}
		  	\def\pad{0.5*\dist}
		  	\def\spaces{6}
		  	\def\offy{-2.7}
			\foreach \xa/\ya/\xb/\yb/\p in {%
				0*\dist/+2/1*\dist/0/-1
			} {
				\draw [-] (\xa*\dist+\p*\pad,\ya*\dist-\pad+\offy) --
						  (\xb*\dist,\yb*\dist+\pad+\offy);
			}
			\draw [-] (3.5*\dist+\spaces*\dist+-1*\pad,2*\dist-\pad+\offy)
				edge [->, bend right] (2*\dist+\spaces*\dist,0*\dist+\pad+\offy);
			
		  	\foreach \num/\x/\y/\type in {%
		  		  /-3*\dist/+2/square,%
		  		  /-1*\dist/+2/square,%
		  		  /+1*\dist/+2/square,%
		  		  /+2*\dist+\spaces/+2/square,%
		  		  /+3*\dist+\spaces/+2/square,%
		  		  /+4*\dist+\spaces/+2/square,%
		  		 a/+4.5*\dist+\spaces/+0/ghost square,%
		  		 c/+5.5*\dist+\spaces/+0/ghost square,%
		  		 a/+0/+0/square,%
		  		 c/+1/+0/square,%
		  		 a/1*\dist+\spaces/+0/square,%
		  		 b/2*\dist+\spaces/+0/square,%
		  		 c/3*\dist+\spaces/+0/square%
		  	} {
				\draw (\x*\dist,\y*\dist+\offy) node [\type] {\num};
		  	}
		  	\draw [-to,thick,snake=snake,segment amplitude = 0.4mm,%
		  		segment length=2mm,line after snake=1mm]
		  		(3*\dist,\offy+\dist) -- node[text centered,xshift=0cm]
		  		{\insertop$(b)$} (-\dist+\dist*\spaces,\offy+\dist);
		  \end{scope}
		\end{tikzpicture}
	\end{minipage}
	\hspace{0.02\textwidth}
	\begin{minipage}{0.48\textwidth}
		\centering
		\noindent\textbf{Pruning Deletion} \\
		\begin{tikzpicture}[node distance=1.2cm,>=stealth',bend angle=30,auto]
		  \begin{scope}
		  	\def\dist{0.5}
		  	\def\pad{0.25}
		  	\def\spaces{3}
		  	\def\offy{0}
			\foreach \xa/\ya/\xb/\yb/\p\type in {%
				+00/+0/\pad-2/-2/-1/every path,%
				+01/+0/+00/-2/-1/every path,%
				+02/+0/-\pad+2/-2/-1/every path,%
				+02/+0/-\pad+4/-2/+1/every path,%
				4*\dist/+2/+2*\dist/+0/+1/every path,%
				\spaces+7*\dist/+0/\spaces+4*\dist/-2/-1/ghost path,%
				\spaces+9*\dist/+0/\spaces+7*\dist/-2/-1/ghost path,%
				\spaces+11*\dist/+0/\spaces+10*\dist/-2/-1/ghost path,%
				\spaces+13*\dist/+0/\spaces+14*\dist/-2/-1/ghost path%
			} {
				\draw [-] (\xa*\dist+\p*\pad,\ya*\dist-\pad+\offy) edge [\type]
						  (\xb*\dist,\yb*\dist+\pad+\offy);
			}
			
		  	\foreach \num/\x/\y/\type in {%
		  		  /+3*\dist/+2/square,%
		  		  /+5*\dist/+2/square,%
		  		  /+7*\dist/+2/square,%
		  		  /\spaces+10*\dist/+2/square,%
		  		  /\spaces+12*\dist/+2/square,%
		  		  /\spaces+14*\dist/+2/square,%
		  		 b/+00/+0/square,%
		  		 c/+01/+0/square,%
		  		 d/+02/+0/square,%
		  		 /\pad-2/-2/square,%
		  		 b/+00/-2/square,%
		  		 /-\pad+2/-2/square,%
		  		 e/-\pad-0.5*\dist+4/-2/square,%
		  		 f/-\pad+0.5*\dist+4/-2/square,%
		  		 b/\spaces+0.0+7*\dist/+0/ghost square,%
		  		 c/\spaces+0.5+7*\dist/+0/ghost square,%
		  		 d/\spaces+1.0+7*\dist/+0/ghost square,%
		  		 /\spaces+\pad+3*\dist/-2/ghost square,%
		  		 b/\spaces+7*\dist/-2/ghost square,%
		  		 /\spaces-\pad+11*\dist/-2/ghost square,%
		  		 e/\spaces-\pad-\pad+15*\dist/-2/square,%
		  		 f/\spaces+15*\dist/-2/square%
		  	} {
				\draw (\x*\dist,\y*\dist+\offy) node [\type] {\num};
		  	}
		
			\draw [-,bend angle=25] (\spaces+13*\dist*\dist+-1*\pad,1-\pad)
				edge [->,bend left] (\spaces+14*\dist*\dist,-1*\dist-\pad);
			
		  	\draw [-to,thick,snake=snake,segment amplitude = 0.4mm,%
		  		segment length=2mm,line after snake=1mm]
		  		(4.5*\dist,\offy) -- node[text centered,xshift=0cm]
		  		{\deleteop$(b)$} (\spaces+0.5*\dist,\offy);
		  \end{scope}
		\end{tikzpicture}

        \vspace{3mm}
		\noindent\textbf{Simple Deletion} \\
		\begin{tikzpicture}[node distance=1.2cm,>=stealth',bend angle=10,auto]
		  \begin{scope}
		  	\def\dist{0.5}
		  	\def\pad{0.5*\dist}
		  	\def\spaces{6}
		  	\def\offy{0}
			\foreach \xa/\ya/\xb/\yb/\p in {%
				0*\dist/+2/2*\dist/0/-1
			} {
				\draw [-] (\xa*\dist+\p*\pad,\ya*\dist-\pad+\offy) --
						  (\xb*\dist,\yb*\dist+\pad+\offy);
			}
			\draw [-] (3.5*\dist+\spaces*\dist+-1*\pad,2*\dist-\pad+\offy)
				edge [->, bend right] (\spaces-9*\dist*\dist,0*\dist+\pad+\offy);
			
		  	\foreach \num/\x/\y/\type in {%
		  		  /-3*\dist/+2/square,%
		  		  /-1*\dist/+2/square,%
		  		  /+1*\dist/+2/square,%
		  		  /+2*\dist+\spaces/+2/square,%
		  		  /+3*\dist+\spaces/+2/square,%
		  		  /+4*\dist+\spaces/+2/square,%
		  		 a/+3.5*\dist+\spaces/+0/ghost square,%
		  		 b/+4.5*\dist+\spaces/+0/ghost square,%
		  		 d/+5.5*\dist+\spaces/+0/ghost square,%
		  		 a/+00/+0/square,%
		  		 b/+01/+0/square,%
		  		 d/+02/+0/square,%
		  		 a/1*\dist+\spaces/+0/square,%
		  		 d/2*\dist+\spaces/+0/square%
		  	} {
				\draw (\x*\dist,\y*\dist) node [\type] {\num};
		  	}
		  	\draw [-to,thick,snake=snake,segment amplitude = 0.4mm,%
		  		segment length=2mm,line after snake=1mm]
		  		(3.5*\dist,\dist) -- node[text centered,xshift=0cm]
		  		{\deleteop$(b)$} (-0.5*\dist+\dist*\spaces,\dist);
		  \end{scope}
		\end{tikzpicture}
	\end{minipage}
\end{center}
\caption{the four \kst\ update operations.
	}
\label{fig-all-updates}
\end{figure}
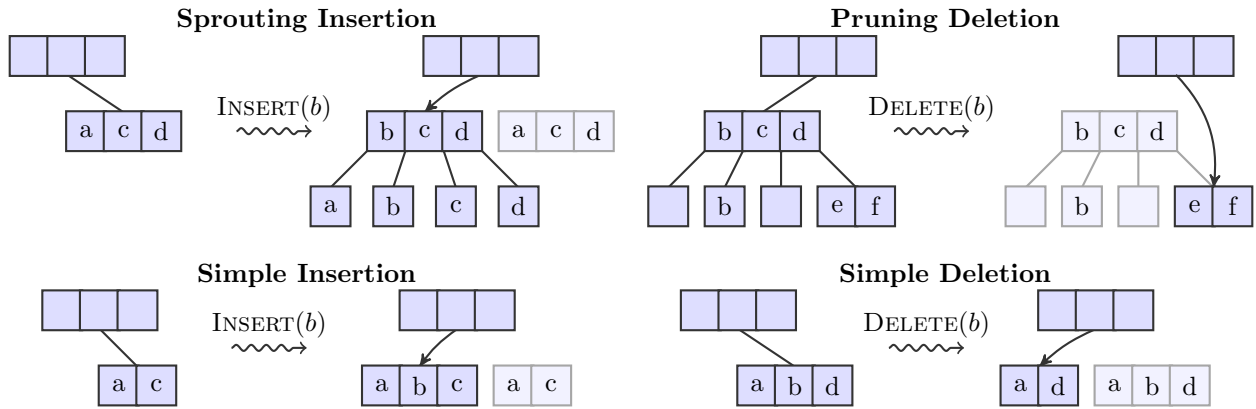

These operations are implemented as in the work of Brown and Helga~\cite{BH11:opodis}.
\find\ is conceptually simple, and extremely fast.
It 
traverses the tree just as it would in the sequential case.
However, concurrent updates can prevent its termination (see \cite{EFRB10:podc}).
\ins\ and \del\ are each split into two cases,
for a total of four update operations:
sprouting insertion, simple insertion, pruning deletion and simple deletion
(see Figure~\ref{fig-all-updates}). 
%
An \ins\ into leaf $l$ will be a sprouting insertion
when $l$ already has $k-1$ keys.
Since $l$ cannot accommodate any more keys,
it is atomically replaced (using CAS)
by a newly created sub-tree that contains
the original $k-1$ keys, as well as the key being inserted.
Otherwise, the \ins\ will be a simple insertion,
which will atomically replace $l$ with a newly created leaf,
containing the keys of $l$, as well as the key being inserted.
A \del\ from leaf $l$ will be a pruning deletion
if $l$ has one key and exactly one non-empty sibling.
Otherwise, the \del\ will be a simple deletion,
which will atomically replace $l$ with a newly created leaf,
containing the keys of $l$, except for the key being deleted.
Note that if $l$ has one key, then $l$ will be empty after a simple deletion.
With this set of updates, it is easy to see that the keys of a node never change,
and that internal nodes always have $k-1$ keys and $k$ children.

Non-blocking progress is obtained using a helping scheme that is generalized from the work of Ellen et~al. \cite{EFRB10:podc}, and is somewhat similar to the cooperative technique of Barnes \cite{Barnes}.
Every time a process $p$ performs an operation $O$, it stores information in the nodes it will modify to allow any other process to perform $O$ on $p$'s behalf.
Specifically, $p$ constructs a \textit{descriptor} $d$ for $O$ that describes how a process can perform $O$, and then invokes an idempotent \textsc{Help} procedure on $d$ to complete $O$.
This \textsc{Help} procedure stores a pointer to $d$ in each node that $O$ will modify.
These pointers are used to coordinate operations so they do not interfere with one another.
Once a node points to descriptor $d$, it will continue to point to $d$ until $O$ has completed.
As long as a node points to descriptor $d$, we say the node is \textit{protected} by $O$.
An operation can only modify nodes that are currently protected by it.
Therefore, another operation cannot modify a node that is protected by $O$ until $O$ has completed.
Whenever $p$ is prevented from making progress by another operation, $p$ helps that operation complete by invoking \textsc{Help} on its descriptor, and then retries its own operation.

We briefly consider which nodes must be protected by each operation to guarantee that operations do not interfere with one another.
Clearly, each operation must protect all nodes that it will modify.
Additionally, each operation must protect any internal nodes that it will remove from the tree.
Otherwise, for example, a new leaf might be inserted as a child of an internal node that is about to be removed (causing the leaf to be erroneously removed, as well).
Note that leaves never need to be protected, they only contain keys, which do not change.

\section{Range Queries in a \kst} \label{sec-code}

At a high level, our range query algorithm consists of one or more \textit{attempts}, each of which perform a \textit{collect} phase, followed by a \textit{validation} phase.
In the collect phase, a depth-first search (DFS) traverses the relevant part of the tree, and collects pointers to leaves.
Specifically, it collects a pointer to each leaf that currently contains a key in the specified range, and each leaf that a key in the specified range could be inserted into.
We call these leaves the \textit{collected leaves}.
Then, the validation phase checks whether all of the collected leaves were in the tree at the end of the collect phase.
If so, then, since the keys of nodes never change, the collected leaves contain precisely the keys that were in the specified key range at the end of the collect phase.
Otherwise, the keys in the specified range might have changed during the collect or validation phases, so another attempt is performed (i.e., the range query is restarted).

\subsubsection*{Tagging and validation}

\begin{figure}
\codesize
\begin{code}
	\firstline
		\func{SomeUpdate}$(args)$ \nl
	\n		... \nl
			\com removing/replacing a leaf $u$ \nl
			$u.tag := true$ \nl
			remove/replace $u$ \nl
			... \bl\nl
			
	\p  \rquery$(lo, hi)$ \nl
	\n		collect pointers $u_1, u_2, ..., u_n$ to relevant leaves \label{code-collected-start} \nl
			\com validation phase \nl
			for $i = 1..n$ \label{code-collected-validate} \nl
	\n			if $u_i.tag = true$ \tabcom{If an operation has replaced $u_i$, or is about to} \nl
	\n				goto line~8 \tabcom{Restart the range query} \nl
	\p\p	success
\p
\end{code}
\vspace{-4mm}
\caption{Naive approach for performing tagging and validation.}
\label{fig-naive}
\end{figure}

We begin by describing a naive approach for performing efficient validation that fails to guarantee non-blocking progress.
Since the keys of nodes never change, each operation \textit{replaces} any node whose keys it would modify.
Consequently, validation simply needs to check whether any leaf seen during the collect phase was replaced.
To facilitate this, each operation \textit{tags} leaves just before they are replaced.
So, if the validation phase sees that none of the collected leaves are tagged, then they were all in the tree at the end of the collect phase, and we say that validation \textit{succeeded}.
Otherwise, we say that validation \textit{failed}.

Figure~\ref{fig-collected} illustrates how tagging causes validation to fail.
Consider a range query that collects pointers to leaves containing $a$ and $b$ (which are not tagged), and then enters its validation phase.
Before it checks whether these leaves are tagged, \ins$(c)$ tags the leaf containing $b$.
After this, validation is doomed to fail. 
The \ins$(c)$ then changes its child pointer and finishes.

%
To implement tagging, each leaf is augmented with a \textit{tag} bit, which is irrevocably set before an operation removes the leaf from the tree, and validation simply checks the \textit{tag} bits of all collected leaves.
If no leaf is tagged, then validation succeeds.
Otherwise, it fails.
Pseudocode for this approach appears in Figure~\ref{fig-naive}.

It is fairly easy to see how this approach can cause the system to encounter livelock, which means that processes take steps infinitely often, but no operation ever successfully completes.
Consider an execution with two processes, $p$ and $q$.
Suppose $p$ performs an update $O$, and crashes after setting $u.tag := true$ for some leaf $u$, but before removing $u$ from the tree.
Then, $q$ performs a range query which collects a pointer to $u$.
Since $u.tag = true$, the range query will restart.
Furthermore, since $p$ is crashed, $q$ will continue to restart this range query forever.


\subsubsection*{Helping}

One way to ensure progress (for the system as a whole) is to have $q$ \textit{help} operation $O$ complete.
This is easy if $q$ has a pointer to the descriptor for $O$: $q$ can simply invoke the same \textsc{Help} procedure that $p$ was executing when it crashed.
The question is how to obtain a pointer to $O$'s descriptor.

To obtain a pointer to the descriptor for $O$, $q$ could try to find a node protected by $O$.
Since $O$ replaces $u$, it either changes a child pointer of $u$'s parent, or removed $u$'s parent.
In either case, it must protect $u$'s parent, so $u$'s parent will contain a pointer to $O$'s descriptor.
Thus, whenever $q$ collects a pointer to a leaf $v$, it could also save a pointer to $v$'s parent (allowing it to obtain a pointer to $O$'s descriptor later).

Another simple approach is change the \textit{tag} bit into a pointer.
Each time an operation removes a leaf $v$, it must first store a pointer to its descriptor in $v.tag$.
This makes the descriptor of the operation that will remove a leaf immediately available to any process that has a pointer to the leaf.
This approach is more efficient, because it does not add overhead during the collect phase of range queries (and writing a pointer is just as fast as writing a single bit on modern systems).

\begin{figure}[bt]
	\centering
	\includegraphics[width=\textwidth]{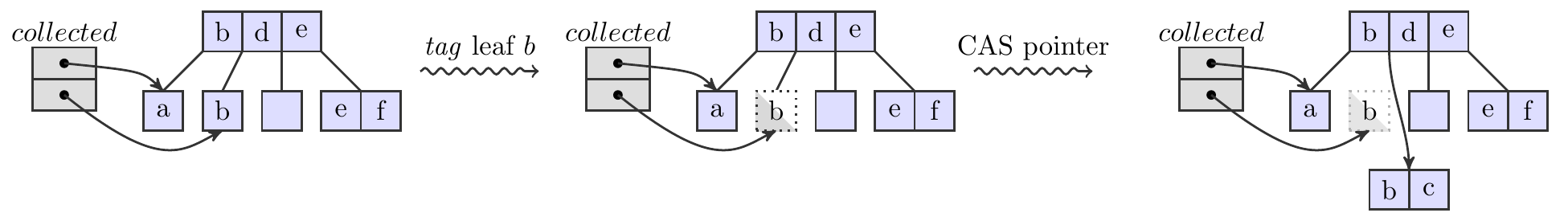}
	\vspace{-4mm}
	\caption{
		Example of how tagging a leaf causes the validation phase of a range query to fail.
	}
	\label{fig-collected}
\end{figure}

Unfortunately, forcing queries to help when they encounter other operations can increase contention in the system.
Hence, our range queries avoid helping.
Luckily, there are other ways to guarantee that the system makes progress without having range queries help other operations.

\subsubsection*{Double checking}

In our algorithm, if a range query sees a node whose \textit{tag} bit is set, it performs some extra work to determine whether the node has actually been replaced (or has simply been tagged by a slow or crashed process).
In this situation, the range query remembers the set of collected nodes 
and then traverses the data structure a second time.
If, after this second traversal, the set of nodes it collects are the same as before, 
then there is a time between the two collects when the nodes collected by both traversals were all in the tree.
If not, then a node was replaced, which implies that another operation completed (so the system is still making progress, as a whole).

This approach is similar to the traditional double-\textsc{Collect} approach, but it offers significant optimizations.
Tagging provides an optimistic way to avoid performing a second traversal of the data structure in the uncontended case.
This reduces the number of nodes that must be accessed by validation, and eliminates the overhead associated with the actual data structure traversal.
Additionally, unlike the double-\textsc{Collect} approach, which performs value-based validation on all keys in the collected nodes, our approach only checks that the set of collected \textit{leaves} does not change.
This greatly reduces overhead in trees with large nodes, containing, e.g., 64 keys.

\subsubsection*{Implementation}

\begin{figure}[ht!]
	{
	
	\codesize
	\begin{minipage}{0.3\textwidth}
		\centering
		\begin{code}
			\firstline
			type \node\ \nl
			\n	 \key\  $\cup\ \{\infty\}$ \ $a_1, ..., a_{k-1}$ \ul
			\p			
		\end{code}%
	\end{minipage}
	\hspace{0.02\textwidth}
	\begin{minipage}{0.3\textwidth}
		\centering
		\begin{code}
			\firstline
			subtype \internal\ of \node\ \nl
			\n	\node\ \ $c_1, ..., c_k$ \ul
			\p
		\end{code}%
	\end{minipage}
	\hspace{0.02\textwidth}
	\begin{minipage}{0.3\textwidth}
		\centering
		\begin{code}
			\firstline
			subtype \leaf\ of \node\ \nl
			\n	
				\info\ \ $tag$ \ul
				\com (initially \nul)
			\p
		\end{code}
	\end{minipage}
		
	\codesize
	\begin{code}
	\firstline
	\func{RangeQuery}(\key\ $lo$, \key\ $hi$) :
				returns List of \node s \nl
	\n		\com \textit{Precondition: $lo, hi \neq \infty$, and $lo \le hi$} \nl
   			Stack $s :=$ new Stack() \nl
            List $collected := null$\nl
            List $prevCollected := null$\bl\nl
	
    \p retry: \label{code-rq-attempt} \nl
	\n		\com DFS traversal to populate $collected$ with \textit{all} leaves that could possibly contain a key in $[lo, hi]$ \nl
            $prevCollected := collected$ \label{code-rq-start} \nl
            $collected :=$ new List() \label{code-rq-declare-collected} \nl
            $s.clear()$ \nl
			$s.push(root.c_1)$ \bl\nl
			
			while $|s| > 0$ \nl
	\n			\node\ $u := s.pop()$ \label{code-rq-pop} \nl
				if $u$ is a \leaf\ then do $collected.append(u)$ \label{code-rq-add-to-collected} \bl\nl
				
				\com Determine which children of $u$ to traverse \nl
				int $l := 1$ \nl
				int $r := k$ \bl\nl
				
				\com Find the rightmost subtree that could contain a key in $[lo, hi]$ \nl
				while $r > 1$ and $hi < u.a_{r-1}$ \label{code-rq-skip-hi} \nl
	\n				$r := r - 1$ \bl\nl
	
	\p			\com Find the leftmost subtree that could contain a key in $[lo, hi]$ \nl
				while $l < k$ and $lo \ge u.a_l$ \label{code-rq-skip-lo} \nl
	\n				$l := l + 1$ \bl\nl
	
	\p			\com perform DFS left to right (so push onto stack from right to left) \nl
				for $i = r .. l$ (iterate backwards from $r$ down to $l$)\label{code-rq-loop-push} \nl
	\n				$s.push(u.c_i)$ \label{code-rq-push} \label{code-rq-traverse-end} \bl\nl
	
	\p\p	\com Validate \nl
            if $prevCollected = null$ then \label{code-rq-validate-start} \nl
    \n          \com Optimistically validate using tags for the \textit{first} traversal \nl
			    if any node in $collected$ is \textit{tagged} then goto retry \label{code-rq-retry1} \nl
    \p  else \nl
    \n          \com Pessimistically validate using double collect for subsequent traversals \nl
                if $collected$ and $prevCollected$ have different contents then goto retry \label{code-rq-retry2} \nl\nl
		
	\p      \com Return all leaves in $collected$ that contain some key in range $[lo,hi]$ \nl
			List $result :=$ new List() \label{code-rq-collect-start}  \nl
			for each $u$ in $collected$ \label{code-rq-loop2} \nl
	\n		if at least one of $u$'s keys is in range $[lo, hi]$ then do $result.append(u)$
			\label{code-rq-collect-end} \bl\nl
			
	\p		return $result$ \label{code-rq-return}
    \p
	\end{code}
	}
	\vspace{-4mm}
	\caption{Pseudocode for \func{RangeQuery}.
			\func{RangeQuery} takes keys $lo$ and $hi$ as arguments
			and returns a list of pointers to all leaves that
				(a) were in the tree at the linearization point, and
				(b) have a key in $[lo, hi]$.
			}%
	\label{code0}
\end{figure}


\paragraph{Range query operations}
Java-like pseudocode for
the \rquery\ operation is given in Figure~\ref{code0}.
We borrow the concept of a \textit{reference} type from Java.
In this psuedocode, variables of any type $E \notin \{$int, boolean$\}$
are references to objects of type $E$.
A reference $x$ is like pointer,
but is dereferenced when a field of the object
is accessed with the ($.$) operator, as in: $x.field$
(which means the same as \texttt{x->field} in C).
References take on the special value \nul\
when they do not point to any object.
(However, no field of any node is ever \nul.)

A \rquery\ operation $R$ starts by declaring two lists, $collected$ and $prevCollected$, to hold pointers to leaves collected in the current traversal, and in the previous traversal (if the traversal is restarted).
Then, it begins its first \textit{attempt} at line~25.
It sets $prevCollected := collected$, then sets $collected$ to a newly created list.
(In the first attempt, this effectively sets $prevCollected := null$. In each subsequent attempt, setting $prevCollected := collected$ saves the leaves found during the previous attempt's collect phase.)
The $collected$ list is then populated by the loop at lines~27-45.
This loop implements depth-first-search (DFS) using a stack (instead of recursion).
For efficiency, it does not traverse subtrees that cannot contain any key in $[lo,hi]$.
The loop at line~38 skips subtrees that contain only keys strictly greater than $hi$.
The loop at line~41 skips subtrees that contain only keys strictly less than $lo$.
Both of these loops use the fact that keys are maintained in increasing order within each node.
The loop at line~44 then pushes the roots of all remaining subtrees onto the stack.
All paths that could lead to keys in $[lo,hi]$ are explored, and all terminal leaves on these paths are placed in $collected$ at line~33.

After the collect phase, $R$ enters a validation phase at line~47.
In the first attempt, $R$ optimistically uses the \textit{tag} bits of collected leaves to determine whether any of them changed during the collect phase.
If no collected leaf is tagged, then \textit{no} collected leaf changed between when it was visited during the collect phase, and when its \textit{tag} bit was checked, and we say that validation is successful.
In this case, we linearize $R$ at the start of the validation phase.
Note that 
all collected leaves were in the tree, and hence represented an atomic snapshot of the keys in the range $[lo, hi]$, when $R$ is linearized.
Otherwise, if some collected leaf is tagged, then $R$ starts another attempt (returning to line~27).
In each attempt after the first, $R$ uses the classical double-collect approach to validate collected leaves, instead of using \textit{tag} bits.
That is, $R$ checks whether $collected$ contains the same leaves as were found by the previous collect phase (in $prevCollected$).
If so, we linearize $R$ between these two collect phases, and we say that validation is successful.
Otherwise, $R$ starts yet another attempt.

If validation is successful, then the loop at lines~56-57 fills a new list, $result$, with pointers to each collected leaf that contains at least one key in $[lo,hi]$.
Since the keys of nodes never change, the range query can simply return pointers to the leaves that contain the keys in the range $[lo, hi]$ (saving the time needed to copy the keys to a separate list).


\newcommand{\helpop}{op}
\begin{figure}[bt]
	\codesize
	\begin{minipage}{0.49\textwidth}
		\centering
		\begin{code}
			\firstline
			\com Type definitions: \nl
			type \node\ \{ \nl
			\n	 final \key\  $\cup\ \{\infty\}$ \ $a_1, ..., a_{k-1}$ \nl
			\p\}
			\nl
			subtype \leaf\ of \node\ \{ 
					\nl
			\n	final int $keyCount$ 
	\end{code}
	\vspace{-4mm}
	
	\begin{graybox}
	\begin{minipage}{\textwidth}
	\codesize
	\begin{code}
	\firstline
	\n	boolean $tag$ \com (initially \false) 
	\end{code}
	\end{minipage}
	\end{graybox}
	\vspace{-6.5mm}
	
	\codesize
	\begin{code}
	\firstline
			\} \p\p \nl
			subtype \internal\ of \node\ \{ 
					\nl
			\n	\node\ \ $c_1, ..., c_k$ \nl
				\info\ \ $\fldinfo$ \ul
				\com (initially a new \clean() object)
						\nl
			\p\}
			\bl\nl
			type \info\ \{ \}\nl
			subtype \iinfo\ of \info\ \{ 
					\nl
			\n	final \node\ \ $l$, $p$, $newChild$\nl
				final int $pindex$\nl
			\p\}
		\end{code}%
	\end{minipage}
	\hspace{0.02\textwidth}
	\begin{minipage}{0.49\textwidth}
		\centering
		\begin{code}
			\firstline
			subtype \dinfo\ of \info\ \{ 
					\nl
			\n	final \node\ \ $l$, $p$, $gp$\nl
				final \info\ \ $p\fldinfo$\nl
				final int $gpindex$\nl
			\p\}
			\nl
			subtype \mk\ of \info\ \{ 
					\nl
			\n final \dinfo\ \ $\fldinfo$\nl
			\p\}
			\nl
			subtype \clean\ of \info\ \{ \} 
					\bl\nl
			\com Initialization:\nl
			shared Internal $root$ \ul
			\label{root-init}
			\ul
			\ul
			\ul
			\ul
			\ul
		\end{code}
	\end{minipage}
	\caption{Type definitions and initialization.}
	\label{code1}
\end{figure}

	\begin{figure}[p!]
	{
	\codesize
	\begin{code}
	\firstline
	\func{Search}(\key\ $key$) :
				$\langle \mbox{\internal}, \mbox{\internal}, \mbox{\leaf}, \mbox{\info}, \mbox{\info}\rangle$ \ul
	\n		\com Used by \func{Insert}, \func{Delete} and \func{Find} to traverse the \kst \ul
			\com \func{Search} satisfies following \textit{postconditions:} \ul
			\com (1) $leaf$ points to a Leaf node, and $parent$ and $gparent$ point to Internal nodes\ul
			\com (2) $parent.c_{pindex}$ has contained $leaf$,
					and $gparent.c_{gpindex}$ has contained $parent$ \ul
			\com (3) $parent.\fldinfo$ has contained $p\fldinfo$, \ul
	\n\n				and $gparent.\fldinfo$ has contained $gp\fldinfo$ \nl
	\p\p   \node\ $gparent$, $parent := root$, $leaf := parent.c1$\label{search-init} \nl 
	       \info\ $gp\fldinfo, p\fldinfo := parent.\fldinfo$\nl
	       int $gpindex$, $pindex := 1$
	\nl
	       while $type(leaf) =$ \internal\ \label{search-loop}
	       	\> \hspace{-1cm} \com Save details for parent and grandparent of leaf \nl
	\n	   $gparent := parent$; $gp\fldinfo := p\fldinfo$\nl
		   $parent := leaf$; $p\fldinfo := parent.\fldinfo$ \label{search-parent}\nl
		   $gpindex := pindex$\nl
		   $\langle leaf, pindex \rangle := \langle $appropriate child of $parent$
			  	by the search tree property%
				,\ul
	\n\n\n\n\n	index such that $parent.c_{pindex}$ is read and stored in $leaf \rangle$ \label{search-child} \nl

	\p\p\p\p\p\p return $\langle gparent, parent, leaf, p\fldinfo, gp\fldinfo, pindex, gpindex \rangle$ \label{search-return} \bl\nl
	\p
	\func{Find}(\key\ $key$) : boolean \nl
	\n   if \leaf\ returned by $\func{Search}(key)$ contains $key$, then return \true, else return \false \label{find-key} \bl\nl
	\p
	\func{Insert}(\key\ $key$) : boolean \nl
	\n \node\ $p$, $newChild$\nl 
	   \leaf\ $l$ \nl
	   \info\ $p\fldinfo$ \nl
	   int $pindex$
	   \nl
	   while \TRUE\ \nl
	   \n $\langle -, p, l, p\fldinfo, -, pindex, - \rangle := \func{Search}(key)$
	   		\label{insert-search} \nl
	      if $l$ already contains $key$ then return \FALSE\ \label{insert-false} \nl
	      if $\typeof(p\fldinfo) \neq$ \clean\ then \label{insert-cleanchk} \nl 
	\n		  \help$(p\fldinfo)$
				\tabcom Help the operation pending on $p$ \label{call-help1} \nl
	\p    else \nl
	\n        if $l$ contains $k-1$ keys \tabcom \textbf{\Icomplex}
						\label{insert-create-start} \nl 
	\n			  $newChild :=$ new \internal\ node with $\fldinfo := new\ \clean()$,\ul
	\n				and with the $k-1$ largest keys in
					$S = \{ key \}$ $\cup$ keys of $l$,\ul
				  	and $k$ new children, sorted by keys, each having
					one key from $S$ \label{newnode1} \nl
	\p\p	  else \tabcom \textbf{\Isimple} \nl
	\n			  $newChild :=$ new \leaf\ node with keys:
					$\{ key \}$ $\cup$ keys of $l$ \label{newnode2}
	\p	  	  \label{insert-create-end} \bl\nl 
		  \iinfo\ $op :=$ new \iinfo$(l, p, newChild, pindex)$ \label{insert-newflag} \nl
		  boolean $result := \CAS(p.\fldinfo, p\fldinfo, op)$ \label{insert-cas}
		  		\tabtabcom \textbf{\ifcas} \bl\nl
		  if $result$ then \tabcom \ifcas\ succeeded \nl
	\n	      \helpinsert$(op)$ \tabcom Finish the insertion
					\label{insert-helpinsert} \nl
		      return \TRUE\ \label{insert-true} \nl
	\p    else \tabcom \ifcas\ failed \nl
	\n		  \help$(p.\fldinfo)$ \tabcom Help the operation pending on $p$
			  		\label{ins-help-after-failure} \label{call-help2}
	\p\p\p\p\bl\nl
	\help(\info\ $\helpop$) \ul
	\n	\com \textit{Precondition:} $op \neq \nul$
		has appeared in $x.\fldinfo$ for some internal node $x$ \nl
	      if $\typeof(\helpop) =$ \iinfo\ then \helpinsert$(\helpop)$ \label{call-HelpInsert}\nl
	      else if $\typeof(\helpop) =$ \dinfo\ then \helpdelete$(\helpop)$ \label{call-HelpDelete}\nl
	      else if $\typeof(\helpop) =$ \mk\ then \helpmarked$(\helpop.\fldinfo)$ \label{call-hm2}
	\p
	\end{code}
	}
	\vspace{-4mm} 
	\caption{Pseudocode for \func{Search}, \func{Find}, \func{Insert} and \func{Help}
			reproduced, without modification, from \cite{BH11:york}.}%
	\label{code2}
	\end{figure}

	\begin{figure}[p!]
	{
	\codesize
	\begin{code}
	\firstline
	\func{Delete}(\key\ $key$) : boolean \nl
	\n \node\ $gp$, $p$ \nl
	   \info\ $gp\fldinfo, p\fldinfo$ \nl
	   \leaf\ $l$ \nl
	   int $pindex$, $gpindex$
	   \nl
	   while \TRUE\ \nl
	\n     $\langle gp, p, l, p\fldinfo, gp\fldinfo, pindex, gpindex \rangle := \func{Search}(key)$
				\label{delete-search} \nl
	       if $l$ does not contain $key$, then return \FALSE\ \label{delete-false} \nl
	       if $\typeof(gp\fldinfo) \neq$ \clean\ then \label{delete-gpcleanchk} \nl
	\n			\help$(gp\fldinfo)$ \label{del-help-unclean-1}
			  	\tabcom Help the operation pending on $gp$ \label{call-help3} \nl
	\p     else if $\typeof(p\fldinfo) \neq$ \clean\ then \label{delete-pcleanchk} \nl
	\n			\help$(p\fldinfo)$ \label{del-help-unclean-2}
				\tabcom Help the operation pending on $p$ \label{call-help4} \nl
	\p     else
					\tabcom Try to flag $gp$\nl
	\n			int $ccount :=$ number of non-empty children of $p$
					(by checking them in sequence) \label{delete-ccount} 
				\nl
				if $ccount = 2$ and $l$ has one key then \tabcom \textbf{\Dcomplex}
						\label{delete-complex} \nl
	\n				\dinfo\ $op :=$ new \dinfo$(l, p, gp, p\fldinfo, gpindex)$
							\label{delete-newcomplexflag} \nl 
					boolean $result = \CAS(gp.\fldinfo, gp\fldinfo, op)$
							\label{delete-complexcas}
							\tabtabcom \textbf{\dfcas}
					\nl 
					if $result$ then
						\tabcom \dfcas\ successful--now delete or unflag \nl
	\n					if \helpdelete(op) then return \TRUE; \label{delete-helpcomplex} \nl 
	\p				else
						\tabcom \dfcas\ failed \nl
	\n					\help$(gp.\fldinfo)$ \tabcom Help the operation pending on $gp$
								\label{call-help5}
	\p				\nl
	\p			else \tabcom{\textbf{\Dsimple}}
							\label{delete-simple} \nl
	\n				\node\ $newChild :=$ new copy of $l$ with $key$ removed \label{newnode3} \nl
			    	\iinfo\ $op :=$ new \iinfo$(l, p, newChild, pindex)$
			    			\label{delete-newsimpleflag} \nl 
					boolean $result := \CAS(p.\fldinfo, p\fldinfo, op)$
						\label{iflag-cas} \label{delete-simplecas}
						\tabtabcom \textbf{\ifcas}
					\nl 
					if $result$ then \tabcom \ifcas\ succeeded \nl
	\n		    		\helpinsert$(op)$ \tabcom Finish inserting the replacement leaf
								\label{delete-helpsimple} \nl
			    		return \TRUE\ \label{delete-true} \nl
	\p              else \tabcom \ifcas\ failed \nl
	\n				\help$(p.\fldinfo)$ \tabcom Help the operation pending on $p$
					\label{call-help6} \label{del-help-after-failure}
	\p\p\p\p\p\p
	\bl
	\nl
	\n\helpdelete(\dinfo\ $op$) : boolean
			\tabcom {\it Precondition}:  $op$ is not \nul \nl
	\n		boolean $result := \CAS(op.p.\fldinfo, op.p\fldinfo,$ new \mk$(op))$
					\tabtabcom \textbf{\mcas}
					\label{helpdelete-markcas}\nl
			\info\ $newValue := op.p.\fldinfo$ \label{helpdelete-readinfo}\nl
			if $result$ or $newValue$ is a \mk\ with $newValue.\fldinfo = op$ then
					\label{helpdelete-markchk} 
				\nl
	\n			\helpmarked$(op)$\label{helpdelete-helpmarked}
					\tabcom Marking successful--complete the deletion\nl
		    return \TRUE \nl
	\p      else \tabcom Marking failed \nl
	\n              \help$(newValue)$ \label{helpdelete-helpother}
						\tabcom Help the operation pending on $p$ \label{call-help7} \nl
			\CAS$(op.gp.\fldinfo, op,$ new \clean$())$\label{helpdelete-backtrack}
						\tabcom Unflag $op.gp$
						\> \com \textbf{\bcas} \nl
				return \FALSE
	\p\p
	\bl
	\nl
	\helpinsert(\iinfo\ $op$) \tabcom {\it Precondition}:  $op$ is not \nul 
	\n		
	\end{code}
	\vspace{-4mm}
	
	\begin{graybox}
	\begin{minipage}{\textwidth}
	\codesize
	\begin{code}
	\firstline
	\n	$op.l.tag := \true$ \p \label{code-helpinsert-tag} 
	\end{code}
	\end{minipage}
	\end{graybox}
	\vspace{-6.5mm}
	
	\codesize
	\begin{code}
	\firstline
	\n		\cas$(op.p.c_{op.pindex}, op.l, op.newChild)$
					\tabcom Replace $l$ by $newChild$
					\> \com \textbf{\ichildcas} \label{helpinsert-caschild}\nl
	\p		\cas$(op.p.\fldinfo, op,$ new \clean$())$
					\tabcom Unflag $p$
					\> \com \textbf{\iucas} \label{helpinsert-unflag}\nl
	\p
	\bl
	\nl
	\helpmarked(\dinfo\ $op$) \tabcom {\it Precondition}:  $op$ is not \nul \nl
	\n	\node\ $other :=$ any non-empty child of $op.p$ different from $op.l$ \ul
	\n		(found by visiting each child of $op.p$),
			or $op.p.c_1$ if none found \label{helpmarked-child} 
	\p\p
	\end{code}
	\vspace{-4mm}
	
	\begin{graybox}
	\begin{minipage}{\textwidth}
	\codesize
	\begin{code}
	\firstline
	\n for each child $u$ of $op.p$, $u \neq other$ \label{code-helpmarked-tag} \nl
	\n		$u.tag := \true$
	\p 
	\end{code}
	\end{minipage}
	\end{graybox}
	\vspace{-4mm}
	
	\codesize
	\begin{code}
	\firstline
	\n	\cas$(op.gp.c_{op.gpindex}, op.p, other)$
		     	\tabcom Replace $p$ by $other$
		     	\> \com \textbf{\dchildcas} \label{helpmarked-caschild}\nl
	\p	\cas$(op.gp.\fldinfo, op,$ new \clean$())$
		     	\tabcom Unflag $gp$
		     	\> \com \textbf{\ducas} \label{helpmarked-unflag}
	\p
	\p
	\end{code}
	}
	\vspace{-1mm} 
	\caption{Pseudocode for \func{Delete}, \helpdelete, \helpinsert\ and \helpmarked.
			\func{Delete} and \helpdelete\ are reproduced, without modification, from
			\cite{BH11:york}.  \helpinsert\ (\helpmarked) is modified to tag
			the leaf being replaced (deleted) just before
			performing the \ichildcas\ (\dchildcas) step.}
	\label{code3}
	\end{figure}

\paragraph{Modifications to $k$-ST operations}

The modifications to the $k$-ST operations are extremely simple.
Pseudocode appears in Figure~\ref{code1}, Figure~\ref{code2} and Figure~\ref{code3}.
Almost all of the code is unchanged from the work of Brown and Helga \cite{BH11:opodis}.
A total of five lines (highlighted in gray) are added to support range queries.
These lines add a $tag$ bit to each leaf, which is initially \false, and set the $tag$ of a leaf to \true\ 
just before an operation removes the leaf from the tree.

\section{Correctness} \label{sec-correctness}


We now give a simple proof of correctness.
We start by 
arguing that our modifications to the $k$-ST operations do not affect their correctness, 
and neither does the introduction of our \rquery\ operation.
%
Next, we show that range queries return correct (linearizable) results.
Finally, we show that the complete implementation supporting \func{Insert}, \func{Delete}, \func{Find} and \rquery\ is non-blocking.

\subsection{Correctness and progress of \func{Insert}, \func{Delete} and \func{Find}}

Since \func{Insert}, \func{Delete} and \func{Find} never read the $tag$ field of any node, our modifications do not affect correctness for these operations.
Furthermore, since \func{RangeQuery} does not modify shared memory, \rquery\ operations cannot affect the correctness of \func{Insert}, \func{Delete} and \func{Find}.
Thus, it remains only to prove that \rquery\ operations are correct, and that all operations are non-blocking. 

\subsection{Correctness of \rquery}

To prove the correctness of \rquery, we must precisely specify its linearization points, and demonstrate that each invocation of \func{RangeQuery}$(lo, hi)$ returns precisely the leaves in the tree that have a key in the range $[lo, hi]$ when it is linearized. 
We think of each \rquery\ as being divided up into \textit{attempts}, each of which begins at an execution of line~25.
Each attempt starts in a \textit{collect} phase and enters a \textit{validation} phase starting at line~47.
\\

\noindent\textbf{Linearization points}
\begin{compactitem}
\item Each invocation of \rquery\ that terminates after its first attempt is linearized at the start of its (only) validation phase.
\item Each invocation of \rquery\ that terminates after more than one attempt is linearized at the start of its last collect phase.
\end{compactitem}

\vspace{2mm}

Before we prove the main result, we first prove some supporting lemmas.

\begin{obs} \label{obs-nodes-immutable}
The keys of nodes do not change. 
To insert or delete a key, \func{Insert} and \func{Delete} \textbf{replace} one or more existing nodes with newly created nodes.
\end{obs}

\begin{defn} \label{def-search-path}
For any configuration $C$, let $T_C$ be the $k$-ary tree formed
by the child references in configuration $C$. We define the \textbf{search path}
for key $a$ in configuration $C$ to be the unique path in $T_C$ that would be followed by the ordinary sequential $k$-ST search procedure.
\end{defn}

\begin{defn} \label{def-range-of-leaf}
We define the \textbf{range} of a leaf $u$
in configuration $C$ to be the set $R$ of keys such that,
for any key $a \in R$, $u$ is the terminal node on
the search path for $a$ in configuration $C$.
(Consequently, if $u$ is not in the tree,
 its range is the empty set.)
\end{defn}

\begin{defn}
A node whose $tag$ bit contains \true\ is said to be \textbf{tagged}.
\end{defn}

\begin{lem} \label{lem-key-update-in-range-implies-replacement}
%
Let $O$ be an invocation of \func{Insert}$(key)$ or \func{Delete}$(key)$ that returns true.
Consider the leaf in the tree whose range contains $key$ just before $O$ is linearized.
This leaf is tagged before $O$ is linearized.
\end{lem}
\begin{proof}
Suppose $O$ is a \func{Delete}.
(The proof for \func{Insert} is a simplified version of this argument.)
Let $op$ be the \info\ object created by $O$, and $S$ be the last \func{Search} invoked by $O$.
Since $O$ returns true, it must 
perform an invocation of \helpdelete$(op)$ that invokes \helpmarked$(op)$.
This invocation of \helpdelete$(op)$ tags $op.l$ before it performs its \dchildcas.
The \info\ object $op$ is created at line~140 or line~148.
In each case, the leaf $l$ returned by $S$ is stored in $op.l$.

It remains only to prove that the range of $l$ contains $key$ just before $O$ is linearized.
However, this is easy to show, given a few facts:
\begin{compactitem}
	\item{By Corollary~21 of \cite{BH11:york}, $l$ is on the search path for $key$ just before $S$ is linearized}
	\item{By Lemma~7 of \cite{BH11:york}, $S$ is linearized before $O$}
	\item{By Corollary~16 of \cite{BH11:york} 
			$l$ is still in the tree just before $O$ is linearized}
\end{compactitem}
Because of these, Lemma~19 of \cite{BH11:york} implies that $l$ is still on the search path for $key$ when $O$ is linearized.
\end{proof}

\begin{obs} \label{obs-dirty-bit-monotonic}
A tagged leaf always remains tagged.
\end{obs}

Now we are able to prove the main result.

\begin{lem} \label{lem-rq-correct}
Each invocation of \func{RangeQuery}$(lo, hi)$
returns a list containing precisely the set of leaves in the tree
that have a key in the range $[lo, hi]$ at the time
the \func{RangeQuery} is linearized.
\end{lem}
\begin{proof}
Consider any invocation $R$ of \func{RangeQuery}$(lo, hi)$ that terminates. 
Let $cl$ be the $collected$ list created in its last attempt.
$R$ returns pointers to all leaves in the $cl$ that contain a key in $[lo, hi]$.
We start by showing that all leaves in $cl$ are in the tree when $R$ is linearized. 

Before $R$ returns, it must either see that each $u$ in $cl$ is not tagged at line~49, or see that $cl$ contains the same leaves as the $collected$ list created by the previous attempt ($prevCollected$) at line~52.
We proceed by cases.

Case 1: no leaf in $cl$ is tagged.
Let $u$ be any leaf in $cl$, and $I_u$ be the interval in $R$'s last attempt that starts when $u$ is read at line~45 and ends at the execution of line~49 where $R$ sees that is $u$ not tagged.
Leaf $u$ cannot be tagged before the end of $I_u$.
Therefore, by Lemma~\ref{lem-key-update-in-range-implies-replacement}, no \func{Insert}$(key)$ or \func{Delete}$(key)$ can be linearized during $I_u$ if $key$ is in the range of $u$.
Consequently, in $R$'s last attempt, after the last leaf is added to $cl$ (at line~45) and before validation begins (at line~47), every leaf in $cl$ is in the tree. 

Case 2: $cl$ contains the same leaves as $prevCollected$.
In this case, all leaves in $cl$ were in the tree at the beginning of the collect phase of $R$'s last attempt.

We now argue that every leaf which contains a key in $[lo, hi]$ when $R$ is linearized is in $cl$. 
This follows immediately from three facts:
\begin{compactenum}[1.]
	\item{In every configuration $C$,
		the tree is a full $k$-ary tree
		and the subtree rooted at $root.c_1$
		is a $k$-ST.}
	\item{The depth-first search (DFS)
		in the last attempt of $R$ 
		starts at $root.c_1$,
		and adds every leaf it visits to $cl$.}
	\item{The only sub-trees that are not traversed by the DFS
		are those that cannot contain a key in $[lo, hi]$.}
\end{compactenum}
Fact 1 follows from Lemma~14 and Lemma~22 of \cite{BH11:york}.
Fact 2 is immediate from the pseudocode.

We now prove Fact 3.
Consider any execution of \func{RangeQuery}$(lo, hi)$ in which some child $v.c_i$ of a node $v$ is \textit{not} pushed onto the stack.
We consider two cases.
Suppose 
$i > r$ when the loop at line~44 is executed.
In this case, $r$ was decremented from $i$ to $i-1$ earlier at line~38, which means the comparison $hi < v.a_{i-1}$ at that line evaluated to \true.
Since the tree is a $k$-ary search tree, the entire subtree rooted at $v.c_i$ contains keys greater than or equal to $v.a_{i-1}$, which is greater than $hi$ (so, the subtree cannot contain a key in $[lo, hi]$).
Now, suppose 
$i < l$ when the loop at line~44 is executed.
In this case, $r$ was incremented from $i$ to $i+1$ earlier at line~41, which means the comparison $lo \ge v.a_i$ at that line evaluated to \true.
Since the tree is a $k$-ary search tree, the entire subtree rooted at $v.c_i$ contains keys strictly less than than $v.a_i$, which is less than or equal to $lo$ (so, the subtree cannot contain a key in $[lo, hi]$).
\end{proof}

\subsection{Correctness and progress of the complete implementation}

\begin{thm}
The implementation of \func{Insert}, \func{Delete}, \func{Find} and \rquery\ described in Figure~\ref{code0}, Figure~\ref{code1}, Figure~\ref{code2} and Figure~\ref{code3} is correct.
\end{thm}
\begin{proof}
Immediate from the proof in \cite{BH11:york} and Lemma~\ref{lem-rq-correct}.
\end{proof}

\begin{thm}
The implementation of \func{Insert}, \func{Delete}, \func{Find} and \rquery\ described in Figure~\ref{code0}, Figure~\ref{code1}, Figure~\ref{code2} and Figure~\ref{code3} is non-blocking.
\end{thm}
\begin{proof}
    Suppose, to obtain a contradiction, that the claim is false.
    In other words, suppose there is an infinite execution with some suffix $S$ in which operations take steps infinitely often, but no operations complete.
    By inspection of the code, child pointers of nodes are only changed by \ichildcas\ and \dchildcas\ steps.
    Every operation that performs a successful \ichildcas\ or \dchildcas\ terminates after a finite number of steps.
    Thus, since since no operations complete, eventually all child pointers in the tree stop changing.
    Let $S'$ be a suffix of $S$ in which child pointers 
    never change.
    Let $\bar O$ 
    be the set of operations that take steps infinitely often in $S'$.
    If $\bar O$ contains a \func{Find}, \func{Insert} or \func{Delete} operation, then the proof in \cite{BH11:york} implies that some operation in $\bar O$ will eventually complete (yielding a contradiction).
    Thus, $\bar O$ must contain only \rquery\ operations. 
    Since no child pointer ever changes, some \rquery\ $R$ in $\bar O$ will eventually perform two attempts that collect the same set of nodes, and terminate (yielding a contradiction).
\end{proof}

\section{Experiments} \label{sec-exp}

\paragraph{Data structures}

In this section, we present the results of experiments comparing the performance of our \kst\ (for $k = 16, 32, 64$) with Snap, Ctrie, and SL, the non-blocking, randomized skip-list of the Java Foundation Classes.
We used the authors' implementations of Snap and Ctrie.
Java code for our \kst\ is publicly available at \url{http://tbrown.pro}.
All of these data structures implement the dictionary ADT, where \ins$(key)$ returns \false\ if an element with this key is already in the dictionary.

\paragraph{Range query implementations}

For SL, a \rquery\ is performed by executing the method
\texttt{subSet(} \texttt{lo, true, hi, true)},
which returns a reference to an object permitting
iteration over the keys in the structure, restricted to $[lo,hi]$,
and then copying each of these keys into an array.
This does not involve any sort of snapshot,
so SL's range queries are not always linearizable.
For Snap, a \rquery\ is performed by following the same process as SL:
i.e., executing \texttt{subSet}, and copying keys into an array.
However, unlike SL, iterating over the result of Snap's \texttt{subSet} causes a snapshot of the data structure to be taken.
Since keys are ordered by their hashed values in Ctrie,
it is hard to perform range queries efficiently.
Instead, we attempt to provide an approximate lower bound on the
computational difficulty of computing \rquery$(lo, hi)$ for
any derivative of Ctrie which uses the same fast snapshot technique.
To do this, we simply take a snapshot,
then iterate over the first $(hi-lo+1)/2$ keys in the snapshot,
and copy each of these keys into an array.
(We explain below that $(hi-lo+1)/2$ is the expected number of keys returned by a \rquery.)
To ensure a fair comparison with the other data structures,
our \kst's implementation of \rquery\ returns an array of keys.
If it is allowed to return a list of leaves,
its performance improves substantially.

\paragraph{Systems}

Our experiments were performed on two multi-core systems.
The first is a Fujitsu PRIMERGY RX600 S6 with 128GB of RAM and four Intel Xeon E7-4870 processors, each having 10 $\times$ 2.4GHz cores, supporting a total of 80 hardware threads (after enabling hyper-threading).
The second is a Sun SPARC Enterprise T5240 with 32GB of RAM and two UltraSPARC\-T2+ processors, each having 8 $\times$ 1.2GHz cores, for a total of 128 hardware threads.
On both machines, the Sun 64-bit JVM version 1.7.0\_3 was run in server mode, with 512MB minimum and maximum heap sizes.
We decided on 512MB after performing preliminary experiments to find a heap size which was small enough to regularly trigger garbage collection and large enough to keep standard deviations small.
We also ran the full suite of experiments for both 256MB and 15GB heaps.
The results were quite similar to those presented below, except that the 256MB heap caused large standard deviations, and the total absence of garbage collection with the 15GB heap slightly favoured Ctrie.

\begin{figure}[p]
	\centering
	\def\figsize{0.75}
	
	\begin{minipage}{0.48\textwidth}
		\centering
		\textbf{Sun UltraSPARC T2+}\\
		\includegraphics[scale=\figsize]{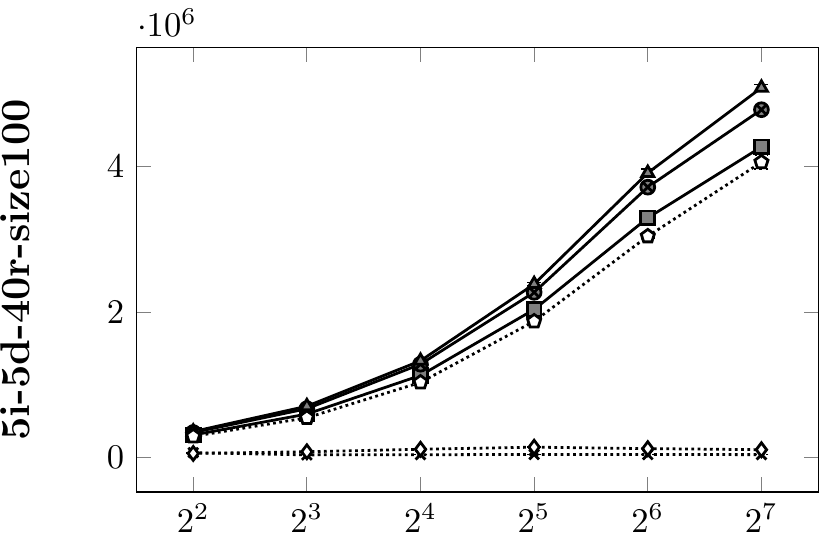}\\
		\includegraphics[scale=\figsize]{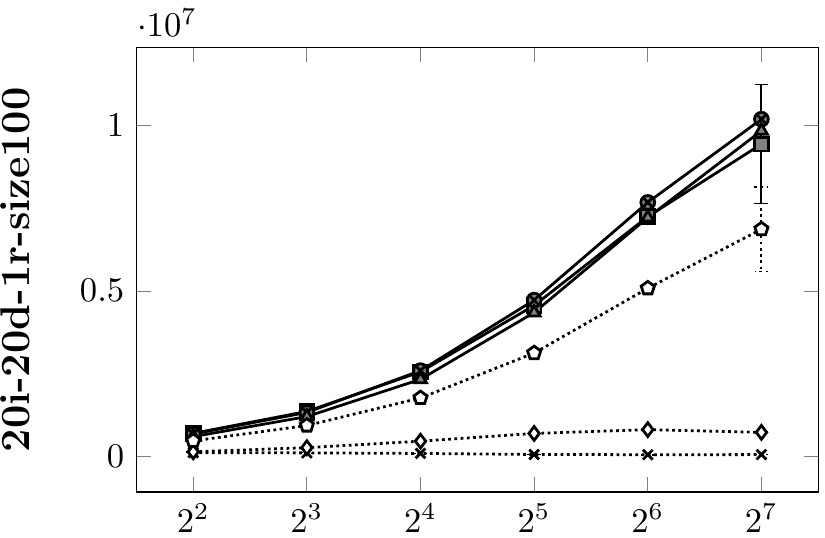}\\
		\includegraphics[scale=\figsize]{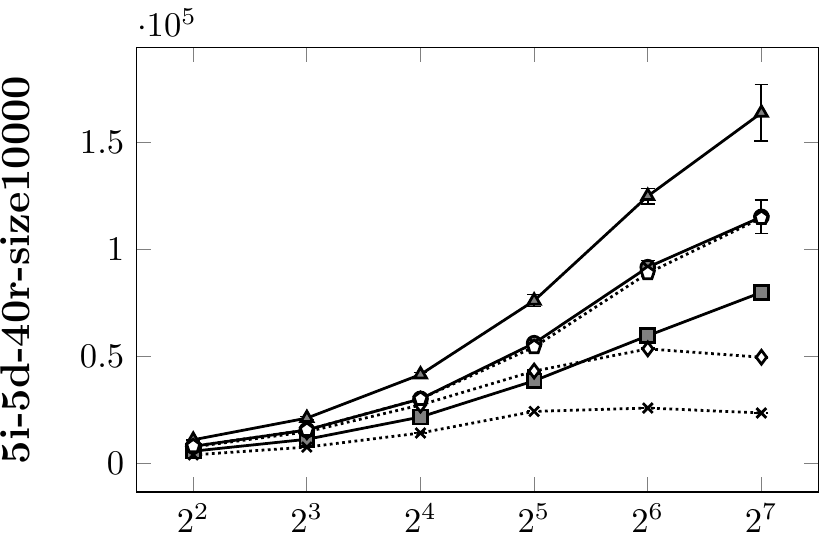}\\
		\includegraphics[scale=\figsize]{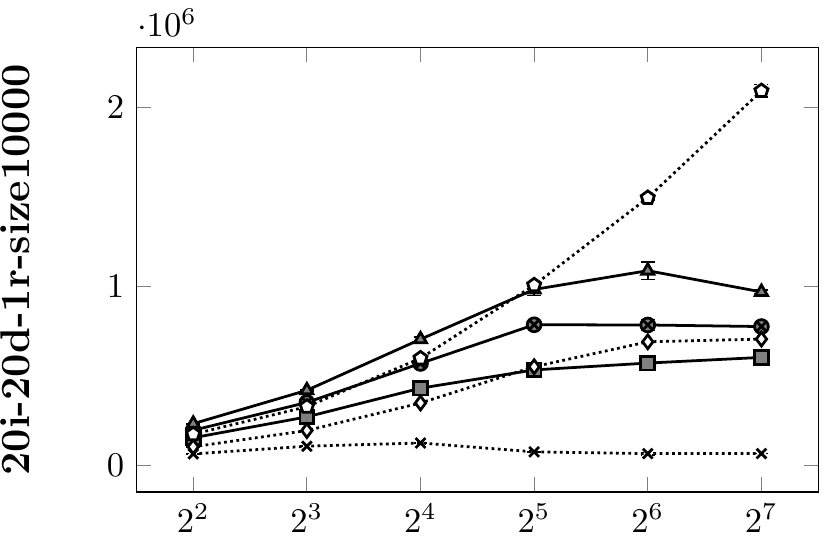}
	\end{minipage}
	\hspace{0.02\textwidth}
	\begin{minipage}{0.48\textwidth}
		\centering
		\textbf{Intel Xeon E7-4870}
		\includegraphics[scale=\figsize]{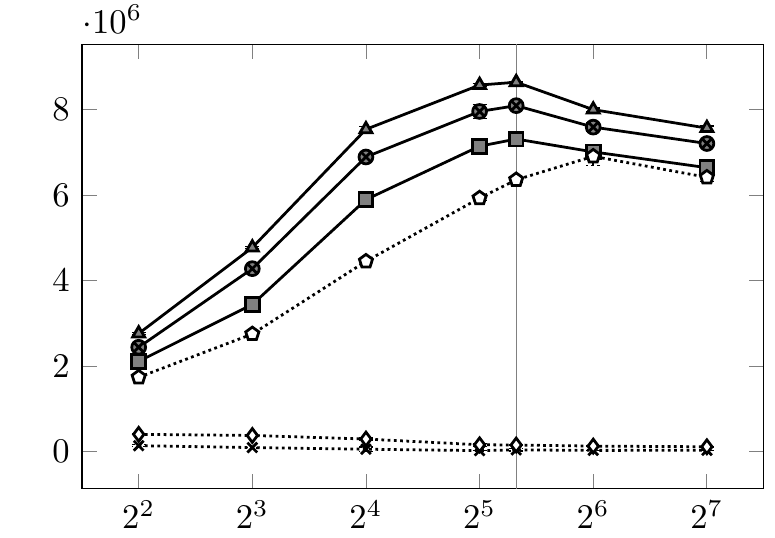}\\
		\includegraphics[scale=\figsize]{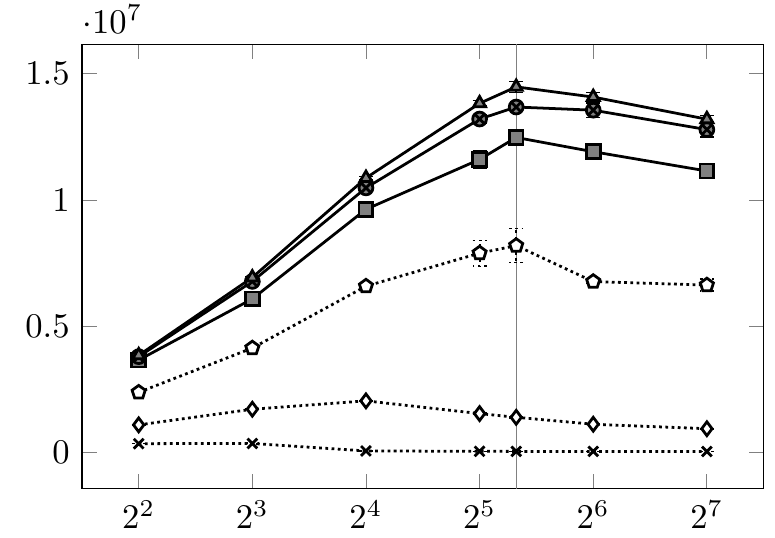}\\
		\includegraphics[scale=\figsize]{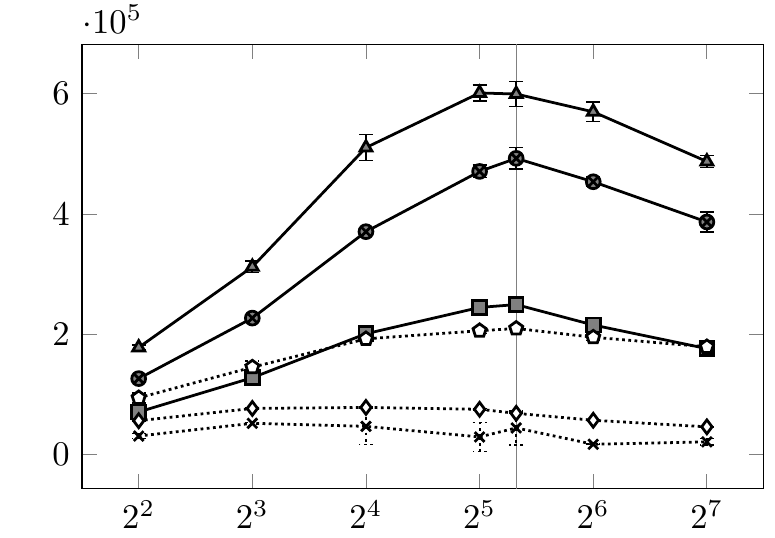}\\
		\includegraphics[scale=\figsize]{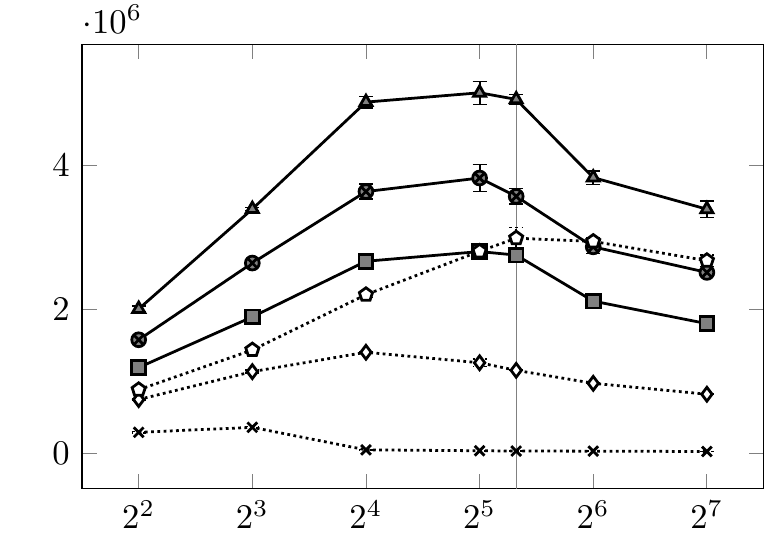}
	\end{minipage}
	
	\vspace{2mm}
	\mbox{\hspace{9.6mm}}
	\includegraphics[scale=0.12]{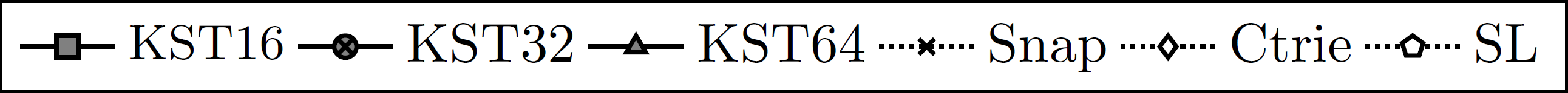}
	\vspace{-2mm}
	\caption{
			Experimental results for various operation mixes (in rows)
			for two machines (in columns).
			The $x$-axes show the number of threads executing,
			and the $y$-axes show throughput (ops/second).
			The Intel machine has 40 cores (marked with a vertical bar).
	}
	\label{fig-exp}
\end{figure}

\begin{figure}[tb]
	\centering
	\def\figsize{0.75}

	\begin{minipage}{0.48\textwidth}
		\centering
		\includegraphics[scale=\figsize]{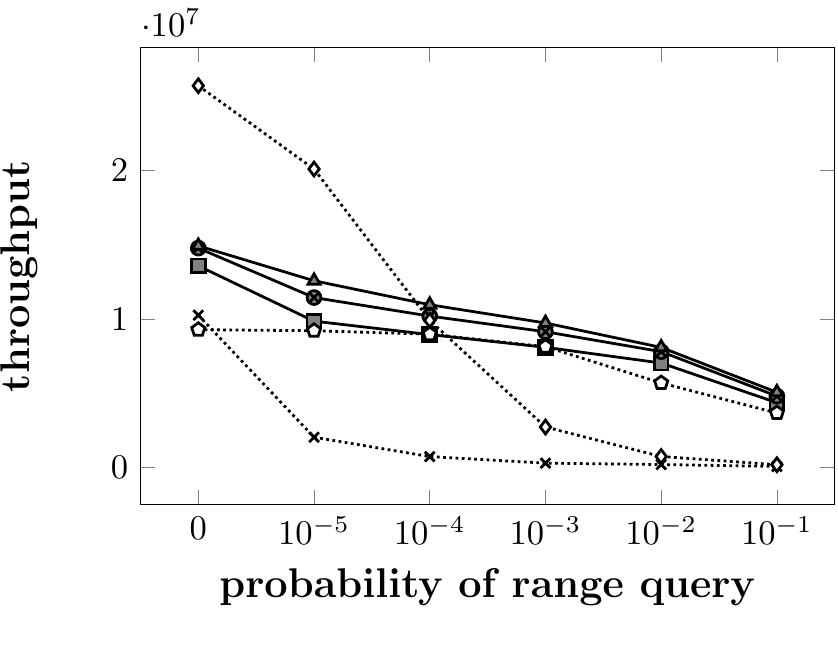}\\
	\end{minipage}
	\hspace{0.02\textwidth}
	\begin{minipage}{0.48\textwidth}
		\centering
		\includegraphics[scale=\figsize]{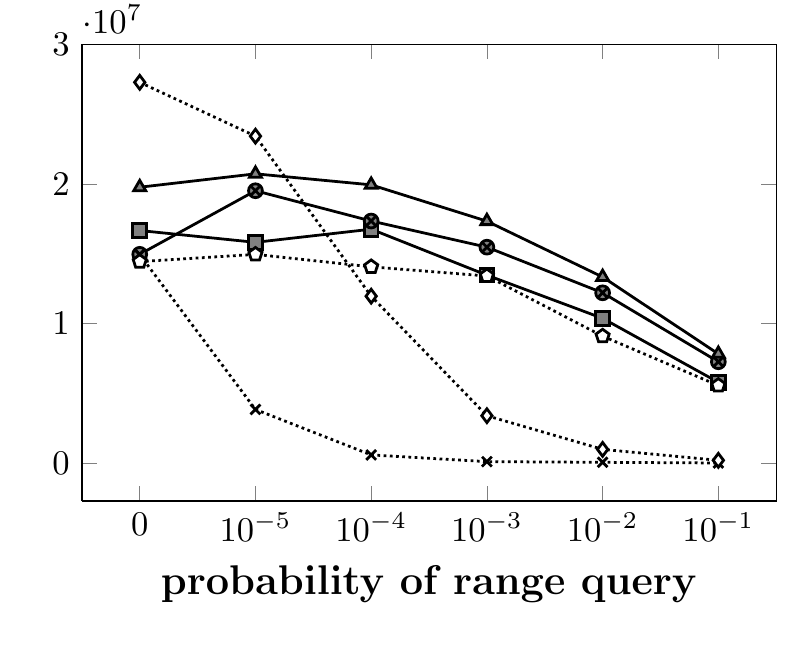}\\
	\end{minipage}

	\mbox{\hspace{9.6mm}
	\includegraphics[scale=0.10]{figures/graphs/legend.png}}
	
	\vspace{-2mm}
	\caption{Sun (left) and Intel (right) results for experiment 5i-5d-?r-100
			wherein we vary the probability of range queries.
			Note: as we describe in Section~\ref{sec-exp},
			Ctrie is merely performing a partial snapshot,
			rather than a range query.
			The Sun machine is running 128 threads,
			and the Intel machine is running 80 threads.
			}
	\label{fig-versus-rq}
\end{figure}
	
\begin{figure}[tb] 
	\centering
	\def\figsize{0.75}
	
	\begin{minipage}{0.48\textwidth}
		\centering
		\includegraphics[scale=\figsize]{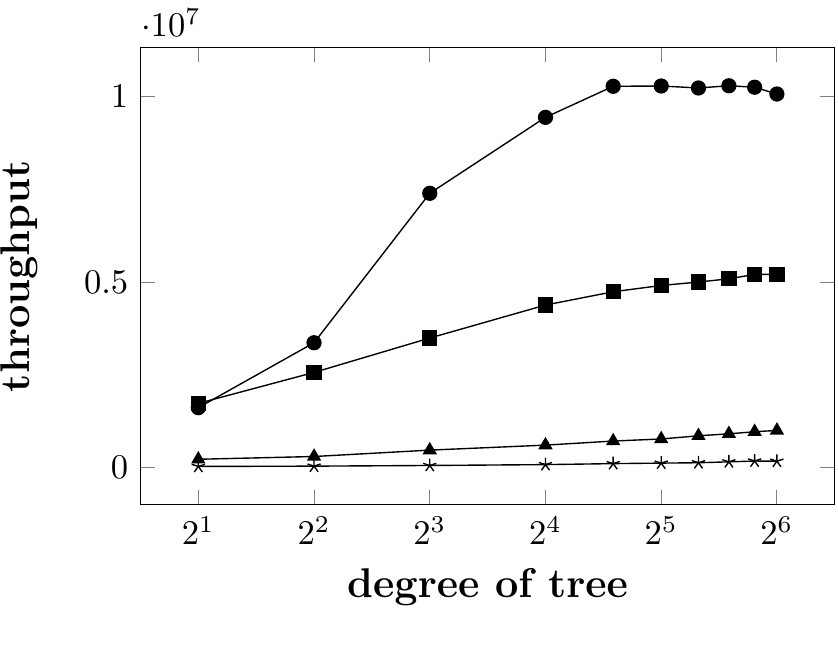}\\
	\end{minipage}
	\hspace{0.02\textwidth}
	\begin{minipage}{0.48\textwidth}
		\centering
		\includegraphics[scale=\figsize]{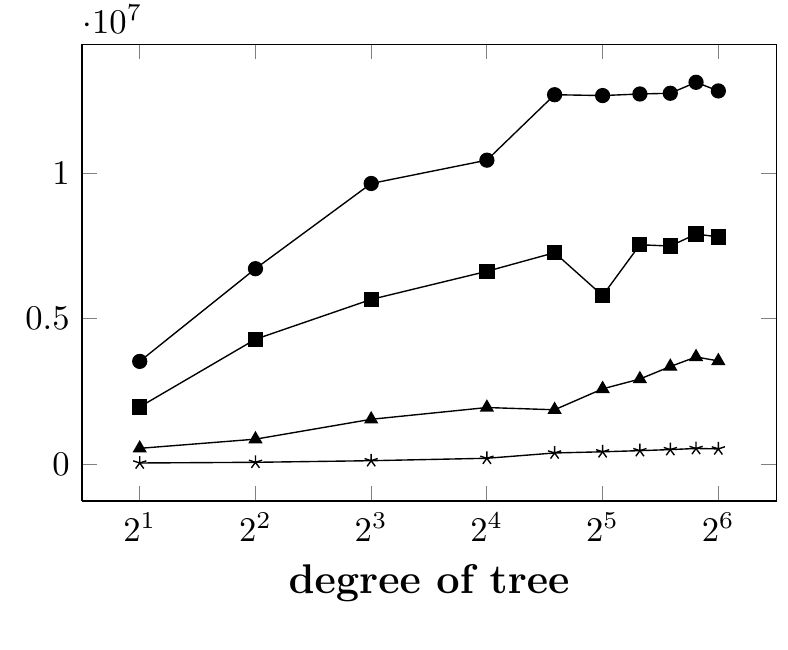}\\
	\end{minipage}
	
	\mbox{\hspace{9.6mm}
	\includegraphics[scale=0.1]{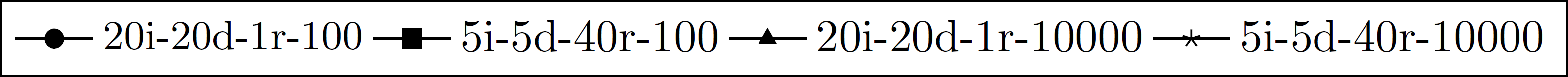}}
	
	\vspace{-2mm}
	\caption{Sun (left) and Intel (right) results showing
			the performance of the \kst\ for many values of $k$,
			and for various operation mixes.
			The Sun machine is running 128 threads,
			and the Intel machine is running 80 threads.}
	\label{fig-versus-arity}
\end{figure}

\paragraph{Methodology}

For each experiment in $\{$5i-5d-40r-size10000, 5i-5d-40r-size100, 20i-20d-1r-size10000, 20i-20d-1r-size100$\}$, each algorithm in $\{$KST16, KST32, KST64, Snap, Ctrie, SL$\}$, and each number of threads in $\{4,8,16,32,64,128\}$, we ran 3 trials, each performing random operations on keys drawn uniformly randomly from the key range $[0,10^6)$ for ten seconds.  Operations were chosen randomly according to the experiment.  Experiment ``$x$i-$y$d-$z$r-size$s$'' indicates $x$\% probability of a randomly chosen operation to be an \ins, $y$\% probability of a \del, $z$\% probability of \rquery$(r,r+s)$, where $r$ is a key drawn uniformly randomly from $[0,10^6)$, and the remaining ($100-x-y-z$)\% probability of a \find.

Each data structure was pre-filled before each trial by performing random \ins\ and \del\ operations, each with 50\% probability, until it stabilized at approximately half full (500,000 keys).
Each data structure was within 5\% of 500,000 keys at the beginning and end of each trial.
This is expected since, for each of our experiments, at any point in time during a trial, the last update on a particular key has a 50\% chance of being an \ins, in which case it will be in the data structure, and a 50\% chance of being a \del, in which case it will not.
Thus, $(hi-lo+1)/2$ is the expected number of keys in the data structure that are in $[lo, hi]$.

In order to account for the ``warm-up'' time an application experiences while Java's HotSpot compiler optimizes its running code, we performed a sort of pre-compilation phase before running our experiments.  During this pre-compilation phase, for each algorithm, we performed random \ins\ and \del\ operations, each with 50\% probability, for twenty seconds.

\paragraph{Results}

Our experiments appear in Figure~\ref{fig-exp}.
The x-axis shows the number of concurrent threads on a logarithmic scale.
Our graphs do not include data for 1 or 2 threads, since the differences between the throughputs of all the algorithms was very small.
Since the Intel machine has 40 cores, we added data points at 40 threads, and drew a vertical bar at 40 threads..
Error bars are drawn to represent one standard deviation.

Broadly speaking, our experimental results from the Sun machine look similar to those from the Intel machine.  
If we ignore the results on the Intel machine for thread counts higher than 40 (the number of cores in the machine), then the shapes of the curves and relative orderings of algorithms according to performance are similar between machines.  A notable exception to this is SL, which tends to perform worse, relative to the other algorithms, on the Intel machine than on the Sun machine.  This is likely due to architectural differences between the two platforms.  Another Intel Xeon system, a 32-core X7560, has also shown the same scaling problems for SL (see \cite{BH11:opodis} technical report).

We now discuss similarities between the experiments 5i-5d-40r-size100 and 20i-20d-1r-size100, which involve small range queries, before delving into their details.
The results from these experiments are highly similar.
In both experiments, all \kst s outperform Snap and Ctrie by a wide margin.
Each range query causes Snap (Ctrie) to take a snapshot, forcing all updates (updates and queries) to duplicate nodes continually.
Similarly, SL significantly outperforms Snap and Ctrie, but it does not exceed the performance of any \kst\ algorithm.
Ctrie always outperforms Snap, but the difference is often negligible.
In these experiments, at each thread count, the \kst\ algorithms either perform comparably, or are ordered KST16, KST32 and KST64, from lowest to highest performance.

Experiment 5i-5d-40r-size100 represents the case of few updates and many small range queries.
The \kst\ algorithms perform extremely well in this case.
On the Sun machine (Intel machine), KST16 has 5.2 times (5.3 times) the throughput of Ctrie at four threads, and 38 times (61 times) the throughput at 128 threads.
The large proportion of range queries in this case allows SL, with its extremely fast, non-linearizable \rquery\ operation, to nearly match the performance of KST16 on the Sun machine.

Experiment 20i-20d-1r-size100 represents the case of many updates and few small range queries.
The \kst s are also strong performers in this case.
On the Sun machine (Intel machine), KST16 has 4.7 times (3.4 times) the throughput of Ctrie at four threads, and 13 times (12 times) the throughput at 128 threads.
In contrast to experiment 5i-5d-40r-size100, since there are few range queries, KST32 and KST64 do not perform significantly better than KST16.
Similarly, with few range queries, the simplicity of SL's non-linearizable \rquery\ operation does not get a chance to significantly affect SL's throughput.
Compared to experiment 5i-5d-40r-size100, the throughput of SL significantly decreases, relative to the \kst\ algorithms.
Whereas KST16 only outperforms SL by 5.2\% at 128 threads on the Sun machine in experiment 5i-5d-40r-size100, it outperforms SL by 37\% in experiment 20i-20d-1r-size100.

We now discuss similarities between the experiments 5i-5d-40r-size10000 and 20i-20d-1r-size10000, which involve large range queries.
In these experiments, at each thread count, the \kst\ algorithms are ordered KST16, KST32 and KST64, from lowest to highest performance.
Since the size of its range queries is fairly large (5,000 keys), Ctrie's fast snapshot begins to pay off and, for most thread counts, its performance rivals that of KST16 or KST32 on the Sun machine.
However, on the Intel machine, it does not perform nearly as well, and its throughput is significantly lower than that of SL and the \kst\ algorithms.
Ctrie always outperforms Snap, and often does so by a wide margin.
SL performs especially well in these experiments, no doubt due to the fact that its non-linearizable \rquery\ operation is unaffected by concurrent updates.

Experiment 5i-5d-40r-size10000 represents the case of few updates and many large range queries.
In this case, SL ties KST32 on the Sun machine, and KST16 on the Intel machine.
However, KST64 outperforms SL by between 38\% and 43\% on the Sun machine, and by between 89\% and 179\% on the Intel machine.
On the Sun machine, Ctrie's throughput is comparable to that of KST16 between 4 and 64 threads, but KST16 outperforms Ctrie by 61\% at 128 threads.
KST64 outperforms Ctrie by between 44\% and 230\% on the Sun machine, and offers between 3.1 and 10 times the performance on the Intel machine.

Experiment 20i-20d-1r-size10000 represents the case of many updates and few large range queries.
On the Sun machine, SL has a considerable lead on the other algorithms, achieving throughput as much as 116\% higher than that of KST64 (the next runner up).
The reason for the \kst\ structures' poor performance relative to SL is two-fold.
First, SL's non-linearizable range queries are not affected by concurrent updates.
Second, the extreme number of concurrent updates increases the chance that a range query of the \kst\ will have to retry.
On the Intel machine, KST64 still outperforms SL by between 21\% and 136\%.
As in the previous experiment, Ctrie ties KST16 in throughput on the Sun machine.
However, KST64 achieves 127\% (270\%) higher throughput than Ctrie with four threads, and 37\% (410\%) higher throughput at 128 threads on the Sun machine (Intel machine).

As we can see from Figure~\ref{fig-versus-rq}, in the total absence of range queries, Ctrie outperforms the \kst\ structures.  However, mixing in just one range query per 10,000 operations is enough to bring it in line with the \kst\ structures.  As the probability of an operation being a range query increases, the performance of Ctrie decreases dramatically.  Snap performs similarly to the \kst\ structures in the absence of range queries, but its performance suffers heavily with even one range query per 100,000 operations.

We also include a pair of graphs in Figure~\ref{fig-versus-arity} for the Intel and Sun machines, respectively, which show the performance of the \kst\ over many different values of $k$, for each of the four experiments.  Results for both machines are similar, with larger values of $k$ generally producing better results.  On both machines, the curve for experiment 20i-20d-1r-size100 flattens out after $k=24$, and 5i-5d-40r-size100 begins to taper off after $k=32$.
Throughput continues to improve up to $k=64$ for the other experiments.
The scale of the graphs makes it difficult to see the improvement in 5i-5d-40r-size10000 but, on the Sun (Intel) machine, its throughput at $k=64$ is 6 times (16 times) its throughput at $k=2$.
This seems to confirm our belief that larger degrees would improve performance for range queries.
Surprisingly, on the Intel machine, experiment 20i-20d-1r-size10000 sees substantial throughput increases after $k=24$.
It would be interesting to see precisely when a larger $k$ becomes detrimental for each curve.

\section{Conclusion} \label{sec-future}

%
%

The \kst\ data structure is unbalanced, so there are pathological inputs that can yield poor performance.  
Additionally, since the \kst\ has no rebalancing operations to combine nodes that contain very few keys, the average number of keys stored in each node can be fairly small in practice (wasting memory and limiting the benefit of increasing $k$).
Recently, Brown implemented lock-free relaxed $(a,b)$-trees~\cite{BrownPhD}, which are balanced trees in which nodes contain between $a$ and $b$ keys, where $b \ge 2a-1$.
We believe that our range query technique can also be used with this implementation, and with several other data structure implementations presented in~\cite{BrownPhD}. 

Another issue is that, in the presence of continuous updates, range queries may starve. 
%
It maybe possible to mitigate this issue by having the \rquery\ operation write to shared memory, and having other updates help concurrent range queries complete.
This direction is left for future research.
Since this work, Petrank and Timnat~\cite{Petrank:2013:LDI:2950115.2950141} have introduced a new technique for implementing \textit{iteration} operations (range queries over the entire universe) that are starvation free, for a certain class of data structures. 
Of course, range queries can be trivially implemented from atomic iteration.
Chatterjee~\cite{Chatterjee:2017:LLR:3007748.3007771} generalized the work of Petrank and Timnat to implement more efficient range queries and to support a larger class of data structures. 
Arbel-Raviv and Brown subsequently introduced a simpler and more efficient technique that supports many more data structures than the preceding approaches~\cite{AB2017}.

Despite the potential for starvation, we believe our present method of performing range queries is practical in many cases.
First, range queries over small intervals involve few nodes, minimizing the opportunity for concurrent updates to interfere.
Second, for many database applications, a typical workload has many more queries (over small ranges) than updates.
For example, consider an airline's database of flights.
Only a fraction of the queries to their database are from serious customers, and a customer may explore many different flight options and date ranges before finally purchasing a flight and updating the database. 


In this work, we described an implementation of a linearizable, non-blocking $k$-ary search tree offering fast searches and range queries.
Our experiments show that, under several workloads, this data structure is the only one with scalable, linearizable range queries.

\subsubsection*{Acknowledgements.}

We would like to thank Faith Ellen for her help in organizing and editing this paper.
We also thank Michael L. Scott at the University of Rochester for graciously providing access to the Sun machine.
This research was supported, in part, by the Natural Sciences and Engineering Research Council of Canada.

\bibliographystyle{abbrv}

\bibliography{bibliography}

\end{document}